\documentclass{LMCS}

\def\dOi{10(3:18)2014}
\lmcsheading%
{\dOi}
{1--26}
{}
{}
{Apr.~\phantom02, 2014}
{Sep.~11, 2014}
{}

\ACMCCS{[{\bf Theory of computation}]: Logic; Semantics and
  reasoning---Program semantics---Denotational semantics}


\usepackage{hyperref}
\usepackage{amsmath,amssymb,amsfonts,latexsym,stmaryrd,mathrsfs,diagrams} 

\newcommand{\uA}{\underline{A}}
\newcommand{\uB}{\underline{B}}

\newcommand{\uR}{\underline{R}}
\newcommand{\uC}{\underline{C}}

\newcommand{\Cc}{\mathcal{C}}
\newcommand{\cO}{\mathcal{O}}

\newcommand{\cF}{\mathcal{F}}

\newcommand{\cV}{\mathcal{V}}
\newcommand{\cE}{\mathcal{E}}
\newcommand{\cP}{\mathcal{P}}
\newcommand{\cH}{\mathcal{H}}
\newcommand{\cS}{\mathcal{S}}
\newcommand{\cU}{\mathcal{U}}

\newcommand{\A}{\mathbb{A}}

\newcommand{\N}{\mathbb{N}}

\newcommand{\B}{\mathbb{B}}

\newcommand{\Si}{\mathbb{S}}
\newcommand{\catC}{\mathsf{C}}
\newcommand{\Hom}{\mathsf{hom}}
\newcommand{\cl}{\mathsf{cl}}

\newcommand{\id}{\mathsf{id}}
\newcommand{\AVal}{\mathbb{A}\mathsf{val}}
\newcommand{\HAVal}{\mathsf{H}\text{-}\mathbb{A}\mathsf{val}}

\newcommand{\dcpo}{{\mathcal{D}\kern-0.14em\mathit{cpo}}}

\renewcommand{\epsilon}{\varepsilon}
\renewcommand{\phi}{\varphi}

\newcommand{\ov}[1]{\overline{#1}}
\newcommand{\wt}{\widetilde}
\newcommand{\wh}{\widehat}
\newcommand{\2}{\Sigma}

\newcommand{\lollipop}{\mbox{$\circ \kern-0.4em \rightarrow$}}

\newcommand{\dbdownarrow}{\rlap{\raise.25ex\hbox{$\shortdownarrow$}}\raise-.25ex\hbox{$\shortdownarrow$}}
\newcommand{\dda}{\rlap{\raise.25ex\hbox{$\shortdownarrow$}}\raise-.25ex\hbox{$\shortdownarrow$}}
\newcommand{\dua}{\rlap{\raise-.25ex\hbox{$\shortuparrow$}}\raise.25ex\hbox{$\shortuparrow$}}
\newcommand{\dbuparrow}{\rlap{\raise-.25ex\hbox{$\shortuparrow$}}\raise.25ex\hbox{$\shortuparrow$}}
\newcommand{\funion}{\mathrel{\makebox[0pt][l]{\hspace{.08em}\raisebox{.4ex}{\rule{.5em}{.1ex}}}\mathord{\cup}}}
\newcommand{\bigfunion}{\mathrel{\makebox[0pt][l]{\hspace{.08em}\raisebox{.4ex}{\rule{.65em}{.1ex}}}\mathord{\bigcup}}}

\newcommand{\dsup}{\mathop{\bigvee{}^{^{\,\makebox[0pt]{$\scriptstyle\uparrow$}}}}}

\DeclareMathOperator{\ua}{\uparrow}
\DeclareMathOperator{\da}{\downarrow\!}

\begin{document}

\title[Observationally-induced algebras in Domain Theory]{Observationally-induced algebras in Domain Theory}

\author[I.~Battenfeld]{Ingo Battenfeld\rsuper a}	
\address{{\lsuper a}Grotenbachstr.~8, Dortmund, D-44225, Germany}	
\email{ingo-battenfeld@email.de}  

\author[K.~Keimel]{Klaus Keimel\rsuper b}	
\address{{\lsuper{b,c}}Fachbereich 4 Mathematik TU Darmstadt\\
         Schlo{\ss}gartenstr.~7, D-64289, Germany}	
\email{\{keimel,streicher\}@mathematik.tu-darmstadt.de}  
\thanks{Work supported by Deutsche Forschungsgemeinschaft (DFG)}

\author[T.~Streicher]{Thomas Streicher\rsuper c}	
\address{\vspace{-18 pt}}	






\begin{abstract}
In this paper we revise and simplify the notion of
\emph{observationally induced} algebra introduced by Simpson and
Schr\"oder for the purpose of modelling \emph{computational effects}
in the particular case where the ambient category is given by
classical domain theory.

As examples of the general framework we consider the various powerdomains.
For the particular case of the Plotkin powerdomain the general recipe leads 
to a somewhat unexpected result which, however, makes sense from a Computer
Science perspective. We analyze this ``deviation'' and show how
to reobtain the original Plotkin powerdomain by imposing further conditions
previously considered by R.~Heckmann and J.~Goubault-Larrecq.  
\end{abstract}

\maketitle

\section*{Introduction}

E.~Moggi in his seminal paper \cite{mogcomplambda} described how to model
\emph{computational effects} via so-called ``computational monads''. Later
Power and Plotkin suggested to model computational effects as free algebras
which after all give rise to monads. Though some effects like continuations 
cannot be modeled this way their approach covers most examples of computational 
effects as described in \cite{plponotions}. In their account the algebras of
interest are specified by equational and inequational laws. As an alternative
A.~Simpson and M.~Schr\"oder in \cite{scsiobservations} suggested to specify
classes of algebras not in an axiomatic way but instead by exhibiting a
prototypical such algebra $\uR$. For such an algebra $\uR$ one may define
a notion of $\uR$-complete algebra. In \cite{scsiobservations,batmfps28,batisdt13,batschrmfps27} it is shown how to construct for every object $X$ of the ambient category 
(of domains) a \emph{free} $\uR$-complete algebra $\uR(X)$ over $X$, called the
\emph{repletion of $X$}.

The notion of $\uR$-complete algebra is defined in analogy with M.~Hyland's
notion of replete object as introduced in \cite{hylsdt} for the purpose of
providing an appropriate notion of completeness within Synthetic Domain Theory.
Actually, in case of an empty signature $\uR$-complete algebras coincide with
$R$-replete objects in the sense of \cite{hylsdt}. The definition of 
$\uR$-completeness in \cite{scsiobservations,batmfps28,batisdt13,batschrmfps27}
appears as somewhat convoluted because categories of algebras are typically 
not cartesian closed and, moreover, in some of the cases considered in 
\emph{loc.cit.}\ the ambient category was not cartesian closed as e.g.\ the
category of topological spaces and continuous maps. We assume our ambient
category to be cartesian closed and the category of algebras to be enriched
over this ambient category. This allows us to come up in section~\ref{sect1} 
with a notion of completeness which is closer in spirit to Hyland's original 
notion of repleteness. Moreover, as shown in section~\ref{sect2} the complete
algebras form a full reflective subcategory when the ambient category are
directed complete partial orders and Scott continuous maps as studied e.g.\
in \cite{abjudomain}.

In the remaining sections we study particular cases of computational effects
corresponding to the various notions of powerdomains. If the computational
effect is given by the computational monad $T$ then the corresponding 
prototypical algebra $\uR$ is chosen to be $T(\2)$ where $\2$ is the 
Sierpi\'nski domain. In most cases $\uR(X)$ turns out to coincide with the 
respective powerdomain of $X$. A notable exception is the Plotkin powerdomain
$\mathcal{P}$ in which case $\A = \mathcal{P}(\2)$ is the 3-element chain
$\bot \sqsubset m \sqsubset \top$ whose elements stand for ``must diverge'',
``may diverge or converge'' and ``must converge'', respectively. It turns
out that for (reasonable) domains $X$ their repletion $\A(X)$ consists of
``formal lenses'' $(C,Q)$ where $C$ is Scott closed in $X$, $Q$ is compact
saturated in $X$ and $C \cap Q$ is nonempty. Obviously, in general such 
formal lenses $(C,Q)$ are not determined by the ``real lens'' $C \cap Q$. 
We also give characterizations of the ``real'' lenses as particular ``formal'' 
lenses using and adapting ideas from
\cite{heckabstvals,jgbprevisions,jglduality}. However, formal lenses
$(C,Q)$ appear as quite natural from a Computer Science 
perspective since $C$ may be understood as a ``safety'' predicate and $Q$ as a
``liveness'' predicate on $X$.

\section{Complete algebras in cartesian closed categories}
\label{sect1}

We assume $\Cc$ to be a cartesian closed category and $\Omega$
a finitary algebraic signature, that is, a collection of operation
symbols $\omega$ each coming with a finite arity $n=n_\omega\in\N$.  
A $\Cc_\Omega$-algebra in $\Cc$ is an object $A$ of $\Cc$ together
with $\Cc$-morphisms $\omega_A\colon A^{n_\omega}\to A$, one for each 
operation symbol $\omega\in\Omega$. We denote $\Cc_\Omega$-algebras by $\uA$,
$\uB$, $\dots$ and by $A$, $B$, $\dots$ the underlying $\Cc$-objects. 
A $\Cc_\Omega$-homomorphism between $\Cc_\Omega$-algebras $\uA$ and $\uB$ is 
a $\Cc$-morphism $h \colon A\to B$ such that $h \circ \omega_A= \omega_B\circ h^n$ 
for every $\omega\in\Omega$ of arity $n$. We denote by $\Cc_\Omega$ the category 
of $\Cc_\Omega$-algebras and $\Cc_\Omega$-homomorphisms.  

As $\Cc$ is cartesian closed, the category $\Cc_\Omega$ has $\Cc$-powers, i.e.\
for $\uA$ in $\Cc_\Omega$ and $X$ in $\Cc$ the power $\uA^X$ is the algebra
whose underlying object is $A^X$ and whose operations are defined pointwise as 
$\omega_{A^X}(\vec{u}) = \lambda x{:}X. \: \omega_A (\vec{u}(x))$. 

Moreover, we assume $\Cc_\Omega$ to be $\Cc$-enriched in the sense that
for $\Cc_\Omega$-algebras $\uA$ and $\uB$ there is an object
$\Hom(\uA,\uB)$ in $\Cc$ with $\Cc_\Omega(\uA,\B^X) \cong
\Cc(X,\Hom(\uA,\uB)$ naturally in $X$. If $\Cc$ has enough  
limits $\Hom(\uA,\uB)$ arises as subobject of the exponential $B^A$ via an 
appropriate equalizer. 

Finally, we fix a $\Cc_\Omega$-algebra $\uR$ as \emph{computational prototype}. 
As originally suggested by A.~Simpson and M.~Schr\"oder \cite{scsiobservations} 
we will define notions of $\uR$-complete algebras which in case of empty signature 
coincide with the notion of an $R$-replete object \cite{hylsdt}. We present two 
such notions which, however, give rise to the same construction when reflecting 
free algebras to the $\uR$-complete ones.

First we give the definitions originally suggested 
in \cite{scsiobservations,batmfps28,batisdt13,batschrmfps27}
but formulated in a way making use of the assumption that $\Cc_\Omega$
is enriched over $\Cc$.

\begin{defi}\label{def:1}{\rm
Let $X$ be an object of $\Cc$ and $\uA$ and $\uC$ be algebras in $\catC$. 
A morphism $e\colon X \to A$ is called \emph{$\uC$-equable}, if the
restriction of the $\Cc$-morphism $C^e \colon C^A \to C^X$ 
to $\Hom(\uA,\uC)$, i.e.\ $C^e \colon \Hom(\uA,\uC) \to C^X \colon 
h \mapsto h \circ e$ is an isomorphism in $\Cc$.}
\end{defi}

Next, we identify our class of computation types as $\uR$-complete 
algebras in the following sense. 

\begin{defi}\label{def:2}{\rm
A \emph{weakly $\uR$-complete algebra} is a $\Cc_\Omega$-algebra $\uC$ such that
every $\uR$-equable morphism  is also $\uC$-equable.
We denote the category of $\Cc_\Omega$-homomorphisms between
weakly $\uR$-complete algebras by $w\Cc_{\uR}$.} 
\end{defi}

Next we give the second stronger version of $\uR$-completeness which has been 
proposed recently by Battenfeld in \cite{batisdt13} (following a suggestion
on p.64 of \cite{sim07talk}) and is closer in spirit to 
Hyland's original definition of $R$-replete objects, see \cite{hylsdt}.
  
\begin{defi}\label{def:1'}{\rm
For $\Cc_\Omega$-algebras $\uA$, $\uB$ and $\uC$, a $\Cc_\Omega$-homomorphism 
$e \colon \uA \rightarrow \uB$ is called \emph{$\uC$-equable}, if the 
$\Cc$-map $\Hom(e,\uC) \colon \Hom(\uB,\uC) \to \Hom(\uA,\uC) \colon
h \mapsto h \circ e$ is an isomorphism in $\Cc$.}  
\end{defi}

\begin{defi}\label{def:2'}{\rm
A $\Cc_\Omega$-algebra $\uC$ is called \emph{$\uR$-complete} if 
every $\uR$-equable homomorphism $e \colon \uA \rightarrow \uB$ 
is also $\uC$-equable. We denote the category of
$\uR$-complete $\Cc_\Omega$-algebras and
$\Cc_\Omega$-homomorphisms by $\Cc_{\uR}$.}  
\end{defi}

We now discuss why $\uR$-completeness is presumably stronger than weak
$\uR$-completeness and why the difference doesn't matter too much
for our purposes. For this, we suppose that for every object $X$ in the 
category $\Cc$ there is an (absolutely) free $\Omega$-algebra over $X$ 
in $\Cc$, i.e.\ a $\Cc_\Omega$-algebra $\cF_\Omega(X)$ together with a 
$\Cc$-map $i_X \colon X \to \cF_\Omega(X)$ such that for every algebra 
$\uA$ the map $\uA^{i_X} \colon \Hom(\cF_\Omega(X),\uA) \to A^X \colon
h \mapsto h \circ i_X$ is an isomorphism in $\Cc$. This amounts to an
internalization to $\Cc$ of the requirement that for every morphism
$e \colon X \to A$ there is a unique homomorphic extension 
$\wh{e} \colon \cF_\Omega(X) \to \uA$ along $i_X$ as depicted in
\begin{diagram}[small]
X & \rTo^{i_X} & \cF_\Omega(X) \\
& \rdTo_{e} & \dTo_{\wh{e}} \\
& & \uA
\end{diagram}
Thus, a morphism $e \colon X \to A$ is $\uC$-equable in the sense of 
Def.~\ref{def:1} iff the homomorphic extension 
$\wh{e} \colon \cF_\Omega(X) \to \uA$ is $\uC$-equable in the sense of
Definition \ref{def:1'}. For this reason, an $\uR$-complete algebra
according to Definition \ref{def:2'} is in particular also
weakly $\uR$-complete in the sense of Definition \ref{def:2}.  

The notion of completeness has the advantage that for interesting
instances of $\Cc$, as e.g.\ the category $\dcpo$ of directed
complete posets, the $\uR$-complete algebras form a full reflective 
subcategory of $\Cc_\Omega$. But, as we will show later, in the category 
$\dcpo$ for every object $X$ the free weakly $\uR$-complete algebra over $X$ 
coincides with the reflection of $\cF_\Omega(X)$ to
the category $\Cc_{\uR}$ of $\uR$-complete algebras.  
 
We recall for later use that the forgetful functor 
$\Cc_\Omega \rightarrow \Cc$ creates limits.  It turns out that 
the same holds for both of the subcategories of weakly $\uR$-complete and 
also $\uR$-complete algebras.

\begin{lem}\label{compcomp}
The forgetful functor $U \colon w\Cc_{\uR} \rightarrow \Cc$ creates limits. 
The same holds for the category of $\uR$-complete algebras.
\end{lem}

\begin{proof} 
Clearly, the forgetful functor factors 
as $\Cc_{\uR} \rightarrow \Cc_\Omega \to \Cc$. The forgetful functor 
$\Cc_\Omega \to \Cc$ is known to create limits. Thus, if $D$ is a diagram 
in $\Cc_{\uR}$, and the limit of $U \circ D$ exists in $\Cc$, then its 
limit $\mathsf{Lim}(D)$ carries a canonical $\Omega$-algebra structure, 
making $(\mathsf{Lim}(D), \{ \omega_{\mathsf{Lim}(D)} \})$ the corresponding 
limit in $\Cc_\Omega$. It only remains to show that 
$(\mathsf{Lim}(D), \{ \omega_{\mathsf{Lim} (D)} \})$ is (weakly) $\uR$-complete, which
follows from a straightforward calculation. 
\end{proof}

\section{Complete algebras in Classical Domain Theory }\label{sect2}

We now study the notions of the previous section for the classical case where
$\Cc$ is the category $\dcpo$ of directed complete partial orders and
Scott continuous maps between them. We start by fixing our notation. 

Recall that a partially ordered set is \emph{directed complete} if
every directed family $(x_i)_{i\in I}$ of elements has a supremum that
we denote by $\dsup_{i\in I}x_i$ in order to indicate that it is the
supremum of a directed family. A map $f$ between dcpos is \emph{continuous} 
(in the sense of Scott) if it preserves the order and suprema of directed 
families. It is well-known that the category $\dcpo$ is cartesian closed, 
complete and cocomplete (see e.g.\ \cite{abjudomain}). 
The exponential, denoted by $Y^X$ and alternatively by $[X\to Y]$, is given
by the set of Scott-continuous functions $u \colon X \to Y$ ordered pointwise. 
Suprema of directed families in $Y^X$ are computed pointwise. A subset $Y$ of 
a dcpo $X$ is said to be a sub-dcpo if, for every directed family 
$(y_i)_{i\in I}$ of elements in $Y$, the supremum $\dsup_{i\in I}y_i$ (taken in $X$) 
belongs to $Y$. 

As in the previous section, we fix a finitary algebraic signature $\Omega$.
A $\dcpo_\Omega$-algebra $\uA$ is a dcpo $A$ together with continuous 
operations $\omega_A\colon A^n\to A$ for every $\omega\in\Omega$ with
arity $n$. A map $h\colon\uA\to \uB$ between $\dcpo_\Omega$-algebras is an 
\emph{$\Omega$-homomorphism} if 
$$h(\omega_A(a_1,\dots,a_n))=\omega_B(h(a_1),\dots,h(a_n))$$ for every
$\omega\in\Omega$ of arity $n$ and all $a_1,\dots,a_n\in A$.  We write 
$\dcpo_\Omega$  for the category of $\dcpo_\Omega$-algebras and 
continuous $\Omega$-homomorphisms. 

For a directed family of continuous $\Omega$-homomorphisms
$\{h_i \colon \underline{A} \rightarrow \underline{B} \}_{i \in I}$, 
their (pointwise) supremum $h(x) = \dsup_{i \in I} h_i(x)$ is again a continuous
$\Omega$-homomorphism $h \colon \uA \rightarrow \uB$. Hence, the continuous 
homomorphisms from $\uA$ to $\uB$ give rise to a sub-dcpo $\Hom(\uA,\uB)$ of 
the exponential $B^A$ in $\dcpo$ for which reason the category $\dcpo_\Omega$ 
is $\dcpo$-enriched. 

For a dcpo $X$ and a $\dcpo_\Omega$-algebra $\uA$, the exponential $A^X$ 
in $\dcpo$ can be endowed with the structure of a $\dcpo_\Omega$-algebra 
by defining the operations $\omega$ on $A^X$ as 
$$\omega_{A^X}(u_1,\dots,u_n) = \lambda x{:}X.\,\omega_A(u_1(x),\dots,u_n(x))$$
where $n$ is the arity of $\omega$. 

It is well-known (see \cite{abjudomain}) that for every dcpo $X$ there is an 
(absolutely) free $\dcpo_\Omega$-algebra over $X$, i.e.\ a $\dcpo_\Omega$-algebra 
$\cF_\Omega(X)$ together with a continuous map $i_X \colon X \to\cF_\Omega(X)$ 
such that for every continuous map $f$ from $X$ to a $\dcpo_\Omega$-algebra $\uA$ 
there is a unique continuous homomorphic extension 
$\widetilde{f} \colon \underline \cF_\Omega(X) \to \uA$ of $f$ along $i_X$ as in 
\begin{diagram}
X&\rTo^{i_X}& \cF_{\Omega}(X)\\
 &\rdTo_{f}&\dTo_{\widetilde{f}}\\
 &       &{\uA}
\end{diagram}
Moreover, the map $f\mapsto\wt f \colon A^X \to \Hom(\cF_\Omega(X),\uA)$ 
is not only bijective but an isomorphism of dcpos (see, e.g.\ \cite{abjudomain}).

We now fix a $\dcpo_\Omega$-algebra $\uR$ as computational prototype and identify 
our class of computation types as $\uR$-complete algebras 
as in Section \ref{sect1}.

\begin{defi}\label{def:1''}{\rm
For $\dcpo_\Omega$-algebras $\uA$, $\uB$ and $\uC$, a continuous homomorphism 
$e \colon \uA \rightarrow \uB$ is called {\emph{$\uC$-equable}}, if every 
continuous homomorphism $h \colon \uA \to \uC$ has a unique continuous homomorphic 
extension $\wh{h} \colon \uB \to\uC$ along $e$ as in
\begin{diagram}
{\uA}&\rTo^{e} &{\uB}\\
   &\rdTo_{h}&\dTo>{\wh{h}}\\
   &       &\underline{C}
\end{diagram}
such that the map 
$$h \mapsto \wh{h} \colon \Hom(\uA,\uC) \to \Hom(\uB,\uC)$$ 
is an isomorphism. } 
\end{defi}

This definition of ``equable'' fits under the general scheme 
of Definition~\ref{def:1'}. First notice that the map 
$\Hom(e,\uC) \colon \Hom(\uB,\uC) \to \Hom(\uA,\uC)$ is always continuous. 
It is surjective if and only if every continuous homomorphism 
$h \colon \uA \to \uC$ has at least one continuous homomorphic extension 
$\wh{h} \colon \uB\to\uC$ along $e$. It is bijective if and only if every  
continuous homomorphism $h\colon \uA \to \uC$ has a unique continuous
homomorphic extension $\wh{h} \colon \uB\to\uC$ along $e$. Then the map 
$\Hom(e,\uC) \colon \Hom(\uB,\uC) \to \Hom(\uA,\uC)$ is an isomorphism 
within $\dcpo$ iff the inverse map $h \mapsto \wh{h}$ preserves the order 
(since isomorphisms in $\dcpo$ are bijective maps which both preserve
and reflect the partial order).
Thus, the map $e$ is {\emph{$\uC$-equable}} if and only if
$\Hom(e,\uC)$ is an isomorphism of dcpos.

\begin{defi}\label{def:2''}{\rm 
A $\dcpo_\Omega$-algebra $\underline{C}$ is said to be
\emph{$\uR$-complete}, if  every $\uR$-equable
$\dcpo_\Omega$-homomorphism $e \colon \uA \rightarrow \uB$ is 
also $\uC$-equable.  
 We denote by $\dcpo_{\uR}$ the category of $\underline{R}$-complete
 $\dcpo_\Omega$-algebras and $\dcpo_\Omega$-homomorphisms. }  
\end{defi}

It is our aim in this section to prove the following theorem.

\begin{thm}\label{th:1}
For every finitary algebraic signature $\Omega$ and every computational
prototype $\uR$, the category $\dcpo_{\uR}$ of all $\uR$-complete
$\dcpo_\Omega$-algebras is a full reflective subcategory of the
category $\dcpo_\Omega$ of all $\dcpo_\Omega$-algebras.
\end{thm}

Recall that for a dcpo $X$ and $\dcpo_\Omega$-algebras $\uA$ and $\uB$ the 
canonical isomorphism $(B^A)^X \cong (B^X)^A$ restricts to a dcpo-isomorphism 
$\Hom(\uA,\uB)^X \cong \Hom(\uA,\uB^X)$
which observation will be useful when proving the following lemma.

\begin{lem}
A Scott-continuous homomorphism $e \colon \uA \to \uB$ 
is $\uC$-equable if and only if it is $\uC^X$-equable for all dcpos $X$. 
\end{lem}

\begin{proof}
Obviously, the backward direction is trivial. For the forward direction
suppose that $e : \uA \to \uB$ is a $\uC$-equable homomorphism and $X$ is 
an object of $\Cc$. First notice that the diagram
\begin{diagram}
\Hom(\uB,\uC)^X & \rTo^{t_B}_\cong & \Hom(\uB,\uC^X) \\
\dTo^{\Hom(e,\uC)^X} & & \dTo_{\Hom(e,\uC^X)}\\
\Hom(\uA,\uC)^X & \rTo_{t_A}^\cong & \Hom(\uA,\uC^X) 
\end{diagram}
commutes where $t_A$ and $t_B$ are the canonical isomorphisms which
``swap arguments''. Since $e$ is $\uC$-equable the left vertical arrow is an 
isomorphism from which it follows that the right vertical arrow is an isomorphism 
as well.
\end{proof}

From the previous lemma it follows that $\uR$-complete algebras are closed 
under arbitrary $\dcpo$-powers. Since $\uR$ clearly is $\uR$-complete,
it follows that all $\dcpo$-powers $\uR^X$ are $\uR$-complete.

\begin{cor}\label{cor:1}
The category $\dcpo_{\uR}$ inherits $\dcpo$-powers from $\dcpo_\Omega$.
\end{cor}

For the next result, let us fix $\dcpo_\Omega$-algebras $\uA$ and $\uB$.  
There is a canonical continuous map $\eta \colon A \rightarrow B^{B^A}$, namely
the transposition of the identity on $B^A$, which in $\lambda$-notation can
be written as $\lambda x.\lambda f. \: f (x)$. Writing $\iota$ for the 
inclusion of $\Hom(\uA,\uB)$ into $B^A$ we can define the map
$B^{\iota} \circ \eta : A \rightarrow B^{\Hom(\uA,\uB)}$ which we also denote
by $\eta$ and in $\lambda$-notation can be written as 
$\lambda x.\lambda h.\: h (x)$.  

\begin{lem}\label{eta}
For all $\dcpo_\Omega$-algebras $\uA$ and $\uB$, the map
$\eta \colon  \uA \rightarrow \uB^{\Hom(\uA,\uB)}$ is an $\Omega$-homomorphism. 
\end{lem}

\begin{proof}
For $\omega\in\Omega$ with arity $n$ we have
$$\begin{array}{lll}
\eta(\omega_A(a_1,\dots,a_n)) 
& = & (\lambda x.\lambda h. \: (h(x))(\omega_A(a_1,\dots,a_n))\\ 
&=&   \lambda h. \: h(\omega_A(a_1,\dots,a_n))\\
&=&   \lambda h. \: \omega_A(h(a_1),\dots,h(a_n))\\
&=&  \omega_{B^{\Hom(\uA,\uB)}}(\lambda h. \: h(a_1),\dots,\lambda
h. \: h(a_n)) \\
&=& \omega_{B^{\Hom(\uA,\uB)}}(\eta(a_1),\dots,\eta(a_n))
\end{array}$$  
for all $a_1,\dots,a_n\in A$.
\end{proof}

Thus, in particular, for every $\dcpo_\Omega$-algebra $\uA$,
we obtain a canonical $\dcpo_\Omega$-homomor\-phism 
$\eta_{\uA} \colon \uA \rightarrow \uR^{\Hom(\uA,\uR)}$. 
Let $\mathcal{J}$ be the collection of all $\uR$-complete 
$\dcpo_\Omega$-subalgebras of $\uR^{\Hom(\uA,\uR)}$ containing the image  
of $A$ under $\eta_{\uA}$.
By Corollary \ref{cor:1}, the algebra $\uR^{\Hom(\uA,\uR)}$ is
$\uR$-complete and hence a member of $\mathcal{J}$, so that
$\mathcal{J}$ is nonempty. We can form the $\dcpo_\Omega$-algebra 
$\uR(\uA)= \bigcap \mathcal J$ and write $r_{\uA} \colon \uA \to \uR(\uA)$ 
for the corestriction of $\eta_{\uA} \colon \uA \to \uR^{\Hom(\uA,\uR)}$.

\begin{prop}\label{prop:factorization}
For every $\dcpo_\Omega$-algebra $\uA$ we have that
\begin{enumerate}
\item $\uR(\uA)$ is $\uR$-complete 
\item $r_{\uA} \colon \uA \to \uR(\uA)$ is $\uR$-equable.
\end{enumerate}
\end{prop}

\begin{proof} 
By construction $\uR(\uA)$ is a $\dcpo_\Omega$-subalgebra of 
$\uR^{\Hom(\uA,\uR)}$ containing the image of $\uA$ under $\eta_A$. 
By Lemma~\ref{compcomp}, $\uR(\uA)$ is $\uR$-complete, since the 
intersection of a collection of $\uR$-complete $\dcpo_\Omega$-subalgebras 
is the limit of the corresponding $\dcpo_\Omega$-subalgebra embeddings. 

It remains to show that $r_{\uA} \colon \uA \rightarrow \uR(\uA)$ 
is $\uR$-equable. For $h \in \Hom(\uA,\uR)$ let $\ell(h) : \uR(\uA) \to \uR$ 
be given by $\lambda f{:}\uR(\uA).\:f(h)$ which is easily seen to be a 
homomorphism. Moreover, the map 
$\ell \colon \Hom(\uA,\uR) \to \Hom(\uR(\uA),\uR)$ is continuous 
since it is given by the $\lambda$-term $\lambda h.\lambda f.\: f(h)$. 
The homomorphism $\ell(h)$ extends $h$ along $r_A$ since 
$(\ell(h) \circ r_A)(a) = \ell(h)(r_A(a)) = \ell(h)(\eta_A(a)) = 
 \eta_A(a)(h) = h(a)$ for all $a \in A$. 

We show now that  $\ell(h): \uR(\uA) \rightarrow \uR$ is the unique 
continuous homomorphic extension of $h$ along $r_{\uA}$.
For this purpose suppose $g,g^\prime \in \Hom(\uR(\uA),\uR)$ with 
$g \circ r_{\uA} = h = g^\prime \circ r_{\uA}$. Then the equalizer 
of $g$ and $g^\prime$ contains the image of $A$ under $r_{\uA}$ and thus 
under $\eta_{\uA}$. Since $\uR(\uA)$ and $\uR$ are $\uR$-complete, the 
equalizer is an $\uR$-complete $\dcpo_\Omega$-subalgebra of $\uR(\uA)$ 
by Lemma~\ref{compcomp}. By construction, $\uR(\uA)$ is the smallest 
$\uR$-complete $\dcpo_\Omega$-subalgebra of $\uR^{\Hom(\uA,\uR)}$ containing 
the image of $\eta_A$. Thus, the equalizer of $g$ and $g^\prime$ must be an 
isomorphism from which it follows that $g = g^\prime$ as desired.

Thus, the map $\ell$ is the continuous inverse of $\Hom(r_{\uA},\uR)$, 
i.e.\ $r_{\uA}$ is $\uR$-equable as claimed.
\end{proof}

\begin{rem}
More generally, for a homomorphism $e : \uA \to \uB$ there is a natural 
correspondence between homomorphisms $h$ making the diagram
\begin{diagram}
\uA   & \rTo^{\eta_A \quad}  & \uR^{\Hom(\uA,\uR)}\\
\dTo^{e} & \ruTo_h &\\
\uB &&
\end{diagram}
commute and continuous sections $s$ of the $\Cc$-map $\Hom(e,\uR) \colon
\Hom(\uB,\uR) \to \Hom(\uA,\uR)$ given by $h(y)(p) = s(p)(y)$ for $y \in B$
and $p \in \Hom(\uA,\uR)$.

Obviously, such $e$ are $\uR$-equable iff $\Hom(e,\uR)$ is monic.
\end{rem}

Having Proposition~\ref{prop:factorization} available 
we can easily give now the

\medskip
\noindent
\emph{Proof} of Theorem~\ref{th:1} :\\
Let $\uA$ be a $\dcpo_\Omega$-algebra. Then $\uR(\uA)$ is $\uR$-complete 
and, since $r_{\uA} \colon \uA \to \uR(\uA)$ is $\uR$-equable, for every 
homomorphism $h$ from $\uA$ to an $\uR$-complete algebra $\uB$ there exists 
a unique homomorphism $\wh{h} : \uR(\uA) \to \uB$ 
with $\wh{h} \circ r_{\uA} = h$.  $\hfill\Box$

\begin{defi}{\rm
For a $\dcpo_\Omega$-algebra $\uA$ we call $\uR(\uA)$ the
\emph{$\uR$-repletion of $\uA$} and $r_{\uA} \colon \uA \to \uR(\uA)$  
the \emph{reflection map} for $\uA$.}
\end{defi}             

One can give the following characterization of $\uR$-repletion 
which avoids any reference to $\uR$-completeness.

\begin{prop}\label{prop:altcharep}
For a $\dcpo_\Omega$-algebra $\uA$, up to isomorphism 
$r_{\uA} \colon \uA \to \uR(\uA)$ is the unique $\uR$-equable homomorphism  
such that for every $\uR$-equable homomorphism $h \colon \uA \to \uB$ 
there is a unique homomorphism
$\wh{h} : \uB \to \uR(\uA)$ with $\wh{h} \circ h = r_{\uA}$.

Alternatively, one may characterize $r_{\uA}$ as the up to isomorphism
unique $\uR$-equable homomorphism from $\uA$ to an $\uR$-complete algebra.
\end{prop}
\begin{proof}
Obviously, the condition of the first characterization determines $r_{\uA}$ 
uniquely up to isomorphism. That $r_{\uA}$ has the required property is an 
immediate consequence of the fact that $\uR(\uA)$ is $\uR$-complete.

The second characterization follows from the fact that $\uR(\uA)$ is
$\uR$-complete and that every $\uR$-equable homomorphism $h$ from $\uA$
to an $\uR$-complete algebra satisfies the condition of the first 
characterization.
\end{proof}

According to the alternative Definition~\ref{def:1} of $\uR$-equability, 
for an algebra $\uA$ a map $e \colon X \to A$ is $\uR$-equable iff every 
continuous map $f \colon X \to R$ has a unique continuous homomorphic 
extension $\wh{f} \colon \uA \to \uR$ such that $f \mapsto \widehat{f} 
\colon R^X \to \Hom(\uA,\uR)$ is continuous. Obviously, a map 
$e \colon X \to A$ is $\uR$-equable in the sense of Definition~\ref{def:1} 
if and only if its homomorphic extension $\wh{e} \colon \cF_\Omega(X) \to \uA$ 
is $\uR$-equable in the sense of Definition~\ref{def:1'}. 

For a dcpo $X$ we consider the map $\eta_X\colon X \to R^{R^X}$ defined 
as $\eta_X(x) = \lambda f.\: f(x)$. Let $w\uR(X)$ be the least weakly
$\uR$-complete $\dcpo_\Omega$-subalgebra of $\uR^{R^X}$ containing 
the image of $\eta_X$ and $r^w_X\colon X\to w\uR(X)$ be the corestriction of 
$\eta_X$.\footnote{Such a subalgebra exists by an argument analogous 
to the one in Proposition~\ref{prop:factorization}.} We call it the weak 
repletion of $X$. Again one can prove that $r^w_X$ is $\uR$-equable. 
Moreover, one easily proves the following analogue of 
Proposition~\ref{prop:altcharep}.

\begin{prop}\label{prop:altcharepweak}
For a dcpo $X$ up to isomorphism $r^w_X \colon X \to w\uR(X)$ is the 
unique weakly $\uR$-equable map such that for every weakly $\uR$-equable 
map $f \colon X \to \uA$ there is a unique homomorphism 
$\wh{f} \colon \uA \to w\uR(X)$ with $\wh{f} \circ f = r^w_X$.

Alternatively, one may characterize $r^w_X$ as the up to isomorphism 
unique weakly $\uR$-equable map to a weakly $\uR$-complete algebra.
\end{prop}      

We do not know whether the weakly $\uR$-complete algebras form 
a full reflective subcategory of $\dcpo_\Omega$. But weak repletion
can be understood as repletion of free algebras as we show next.

\begin{prop}\label{equivrepl}
For every dcpo $X$ the weak repletion $w\uR(X)$ is isomorphic to the
repletion $\uR(\cF_\Omega(X))$ of the free $\dcpo_\Omega$-algebra 
$\cF_\Omega(X)$ generated by $X$. Moreover, the isomorphism 
$h \colon  \uR(\cF_\Omega(X)) \to w\uR(X)$ can be chosen in such a way 
that  $h \circ r_{\cF_\Omega(X)} \circ i_X = r^w_X$.
\end{prop} 

\begin{proof}
Since $\Hom(\cF_\Omega(X),\uR) \to R^X : h \mapsto h \circ i_X$ is an 
isomorphism in $\dcpo$ the maps $i_X$ and $r_{\cF_\Omega(X)} \circ i_X : 
X \to w\uR(\cF_\Omega(X))$ are weakly $\uR$-equable. Moreover, since
$\uR(\cF_\Omega(X))$ is $\uR$-complete it is in particular weakly 
$\uR$-complete. Thus, by Proposition~\ref{prop:altcharepweak} there is 
an isomorphism $h : \uR(\cF_\Omega(X)) \to w\uR(X)$ with 
$h \circ r_{\cF_\Omega(X)} \circ i_X = r^w_X$.
\end{proof}

\section{Free dcpo-algebras}

For semantics, computational effects are mostly modelled by monads
(see \cite{mogcomplambda}). For algebraic effects, monads can be 
specialized to free constructions, an aspect advocated by Plotkin 
and Power \cite{plponotions}. 
In the category of dcpos this comes down to consider $\dcpo$-algebras 
that are free with respect to a collection of equational and 
inequational laws that are considered to be characteristic of
the computational effect under consideration. This approach is quite
different to the observationally-induced approach, to which we want to
compare it in this paper. 

Let $\Omega$ be a finitary signature. If $t_1$ and $t_2$ are two
terms, a $\dcpo_\Omega$-algebra $\uR$ is said to satisfy the
inequational law $t_1\leq t_2$ (resp., the equational law $t_1=t_2$), if
this inequality (resp. equality) holds for every 
instantiation of the variables by elements of $\uR$.    
For a given collection $\cE$ of equational and inequational
laws, denote by $\dcpo_{\Omega,\cE}$ the class of
$\dcpo_\Omega$-algebras that satisfy all the laws in $\cE$. 

It is well-known that over every dcpo $X$ there is a 
\emph{free} $\dcpo_{\Omega,\cE}$-algebra, that is, a
$\dcpo_{\Omega,\cE}$-algebra 
$\cF_{\Omega,\cE}(X)$ together with a continuous map $\iota_X\colon X
\to \cF_{\Omega,\cE}(X)$ such that, for every continuous map $f$ from $X$ 
to some $\dcpo_{\Omega,\cE}$-algebra $\uA$, there is a unique
$\dcpo_\Omega$-homomorphism $\wt f\colon\cF_{\Omega,\cE}(X)\to \uA$
such that $f=\wt f\circ \iota_X$. This is usually proved using the
adjoint functor theorem (see, e.g.\ \cite[Theorem 6.1.2]{abjudomain}).

It is desirable for free algebras to be free in the enriched sense as well,
namely that $A^{\iota_X}\colon\Hom(\cF_{\Omega,\cE}(X),\uA)\to A^X$
is an isomorphism of dcpos. For this it is necessary and sufficient 
that the extension operator $f\mapsto\wt f\colon A^X \to
\Hom(\cF_{\Omega,\cE}(X),\uA)$ preserves the order. 
We do not believe that all free dcpo-algebras are free in this enriched
sense, although we have no counterexample. But restricting the attention to
continuous dcpos and using the description of $\cF_{\Omega,\cE}(X)$
presented in \cite[Section 6.1.2]{abjudomain} one can show:

\begin{prop}
For continuous dcpos $X$, the free 
$\dcpo_{\Omega,\cE}$-algebra $\cF_{\Omega,\cE}(X)$ always is internally free. 
\end{prop}

If $\cF_{\Omega,\cE}(X)$ is internally free, 
then $\iota_X\colon X\to \cF_{\Omega,\cE}(X)$ is $\uR$-equable for 
every $\dcpo_{\Omega,\cE}$-algebra $\uR$. Using the characterization 
of repletion in Proposition \ref{prop:altcharepweak} we obtain
the following criterion for the free algebra construction 
to agree with the observationally induced approach of repletion:

\begin{prop}\label{lem:freereplete}
Let $\uR$ be a $\dcpo_\Omega$-algebra satisfying a collection $\cE$ of
equational and inequational laws. If the free algebra
$\iota_X\colon X\to \cF_{\Omega,\cE}(X)$ is internally free and
$\cF_{\Omega,\cE}(X)$ is $\uR$-complete, then it is (isomorphic to) 
the $\uR$-repletion of $X$.  
\end{prop}

The repletion $r_X\colon X\to \uR(X)$ of a dcpo $X$ with respect to a
computational prototype algebra $\uR$ has some features in common with the
free construction $\iota_X\colon X\to\cF_{\Omega,\cE}(X)$ with respect to
a collection $\cE$ of equational and inequational laws. Indeed, the
repletion $\uR(X)$ is a subalgebra of a dcpo-power of $\uR$. Since
directed suprema in dcpo powers are formed pointwise, $\uR(X)$
satisfies all equational and inequational laws that hold in $\uR$.   
Thus, a first necessary condition for $\uR(X)$ and $\cF_{\Omega,\cE}(X)$
to agree is that all equational and inequational laws that hold in
$\uR$ can be derived from those in $\cE$.  But this condition is by no
means sufficient. It is not easy to illustrate these phenomena. Already
the free algebras are difficult to describe explicitly, while it is
even more difficult to put one's hands on the repletion. 

We will illustrate these phenomena in the relatively simple situation
of the classical powerdomains often named after Hoare, Smyth and
Plotkin which model angelic, demonic and erratic nondeterminism,
respectively. We will see that the Hoare and Smyth powerdomains
agree with the repletion with respect to the natural domains of 
observation,
while there is a big gap between both for the Plotkin powerdomain.
There seems to be an intriguing connection to the fact 
that the inequational theories used in the case of angelic
and demonic nondeterminism are complete in the sense that adding any
inequational law that is not derivable leads to inconsistency; the
inequational theory used in the erratic case has exactly two consistent
extensions, namely the inequational theories used for the angelic and
demonic cases. But we do not want to develop this observation any
further here.

\section{Examples: Powerdomain Constructions}\label{sec:examples} 

We now examine whether the classical powerdomain constructions can
be recovered by the observationally-induced approach. We restrict
ourselves to ordinary nondeterminism (angelic, demonic and erratic)
and compare our approach to the upper (Hoare), lower (Smyth) and 
convex (Plotkin) powerdomains.
In these cases we are indeed in a position to describe the repletion
explicitly which is a rare phenomenon. While in the angelic and
demonic case we obtain the same result as classically, in the erratic
case there is a big gap between the two approaches.   

In a topological setting, the results for the Hoare and Smyth 
powerdomains have been worked out in
\cite{batschrmfps27,batschrobspowersp}. The 
proof strategy is the same in the cases at hand, so we only sketch the
proofs, as the reader should have no problems filling in the details
by consulting \emph{loc.\ cit.} The erratic case will be worked out in
detail. 

A nondeterministic (binary) choice
operator $\funion$ is reasonably supposed to satisfy the following
equational laws:
\[\begin{array}{rcll}
x\funion x& = &x&\text{idempotency}\\
x\funion y &=&y\funion x& \text{commutativity}\\
x\funion (y\funion z)&=&(x\funion y)\funion z& \text{associativity}
\end{array}\]
These laws characterize semilattices. A dcpo together with a
continuous semilattice operation $\funion$ is called a
\emph{dcpo-semilattice}.

One must be careful: In a dcpo-semilattice, $x\funion y$ need not be
the least upper bound nor greatest lower bound of $x$ and $y$ with
respect to the dcpo order $\leq$. A paradigmatic example is
Heckmann's domain 
$\underline{\mathbb{A}}$ \cite{heckabstvals}, the three element chain 
$$\A=\{\bot < m <\top\}$$ with the semilattice operation 
\[
a\funion b = \begin{cases}a &\mbox{ if } a=b\\
                           m  &\mbox{ else }\end{cases} 
\]
where $\top$ can be read as 'must',   $m$ stands for 'may' and  $\bot$
stands for 'impossible'. For constructing the Plotkin
powerdomain in the observationally-induced approach, 
Simpson \cite{sim05talk} has suggested to use Heckmann's domain 
$\underline{\mathbb{A}}$  as computational prototype. 

A dcpo-semilattice in which
$x\funion y$ is the least upper bound of $x$ and $y$ will be called a
\emph{dcpo-join-semilattice} and, in this case, we will usually denote
the semilattice operation by $\vee$. Among dcpo-semilattices the
dcpo-join-semilattices are characterized by the inequational law
$$x\leq x\funion y.\leqno{\rm (J)}$$
 A dcpo-semilattice in which
$x\funion y$ is the greatest lower bound of $x$ and $y$ will be called a
\emph{dcpo-meet-semilattice}\footnote{Let us point out that our
   terminology deviates from the one mostly used in earlier texts,
   where the terms \emph{inflationary} and \emph{deflationary} are used
 for join- and meet-semilattices, respectively.} and, in this case,
we will usually denote the semilattice 
operation by $\wedge$. Among dcpo-semilattices the
dcpo-meet-semilattices are characterized by the inequational law
$$x\funion y\leq x.\leqno{\rm (M)}$$
The paradigmatic examples are obtained from the two element chain $$\2
=\{0<1\}$$ by considering the operation $x\vee y =\max(x,y)$ in the
first and the operation $x\wedge y=\min(x,y)$ in the second case. We
will write $$\2_\vee =(\2,\vee) \quad \text{and} \quad \2_\wedge=(\2,\wedge)$$ for the dcpo $\2$ considered as a 
join-semilattice and meet-semilattice, respectively. 

Semantically, dcpo-join-semilattices are modelling \emph{angelic}
nondeterminism: having a choice is considered to be preferable to an
optimistic observer, since it opens the possibility to make the best
choice;  dcpo-meet-semilattices are modelling  \emph{demonic} 
nondeterminism: having a choice is considered to be undesirable 
to a pessimistic observer, since it opens the possibility
to make the worst possible choice. General dcpo-semilattices are
combining the angelic with the demonic point of view. 

Indeed, if $\uA$ is a dcpo-join-semilattice and $\uB$ a
dcpo-meet-semilattice, the direct product $\uA\times \uB$ with the
operation $(x_1,y_1)\funion(x_2,y_2) = (x_1\vee x_2,y_1\wedge y_2)$ is
a dcpo-semilattice which is neither a dcpo-join-  nor a
dcpo-meet-semilattice. Note that in particular $\A$ may be viewed as a
subsemilattice of $\2_\vee\times \2_\wedge$ with the identifications
$\bot =(0,0)$, $m=(1,0)$ and $\top=(1,1)$.  

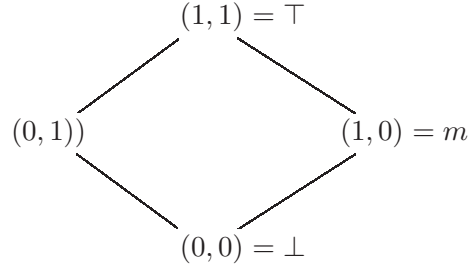
\begin{diagram}[small]
&&\phantom{a}\ \ \ \ \ \ (1,1)=\top&&\\
&\ruLine&&\luLine&\\
(0,1))&&&&\phantom{a}\ \ (1,0)=m\\
&\luLine&&\ruLine&\\
&&\phantom{a}\ \ \ \ \ \ (0,0)=\bot&&
\end{diagram}




As one does not want to ask for any more than just the equational
and inequational laws indicated above, the classical powerdomains over
a dcpo $X$ are defined to be the 
free dcpo-semilattice $\mathcal{P}(X)$, the \emph{Plotkin powerdomain},
the free dcpo-join-semilattice $\cH(X)$, the \emph{Hoare powerdomain},
and the free 
dcpo-meet-semilattice $\cS(X)$, the \emph{Smyth powerdomain}. 

For applying the previously developed concepts to powerdomains, we
place ourselves in the situation where the signature $\Omega$ consists 
of a single binary operation symbol. The $\dcpo_\Omega$-algebras are 
simply dcpos with a continuous binary operation. The dcpo-semilattices 
form a full subcategory.
 
For the category of dcpo-semilattices we are in an exceptional situation:
For two Scott-continuous semilattice homomorphisms 
$f,f'\colon \uA \to \uB$, we may define $f\funion f'$ pointwise 
by $(f \funion f')(x) = f(x) \funion f'(x)$ for all $x \in \uA$ and 
we see that $f \funion f'$ is again a continuous semilattice
homomorphism. Thus, the dcpo of all continuous semilattice homomorphisms 
$f \colon \uA\to \uB$ is a dcpo-subsemilattice of $\uB^A$; we denote it 
by $\Hom(\uA,\uB)$. If $\uB$ has a bottom or a top element, the same holds 
for $\Hom(\uA,\uB)$ -- consider the constant functions with value top and 
bottom, respectively.
If $\uA$ and $\uB$ both have a bottom and a top element, we may restrict
our attention to the dcpo-subsemilattice $\Hom_{0,1}(\uA,\uB)$ of all
continuous semilattice homomorphisms that preserve bottom and top. We note 
that $\Hom_{0,1}(\uA,\uB)$ will not have a bottom or a top element, in general.

\begin{lem}\label{lem:complete}
For every dcpo-semilattice $\uR$ with bottom and top and every dcpo $X$, 
the dcpo-semilattices $\Hom(\uR^X,\uR)$ and  $\Hom_{0,1}(\uR^X,\uR)$ are
$\uR$-complete.   
\end{lem}

Indeed, $\uR^{R^X}$ is $\uR$-complete by Cor.~\ref{cor:1}. 
The dcpo-subsemilattices $\Hom(\uA,\uB)$ and $\Hom_{0,1}(\uA,\uB)$ 
are carved out from $\uB^A$ as intersections of equalizers of homomorphisms
(which to check requires the associativity, commutativity and idempotence
of $\funion$). Thus $\Hom(\uA,\uB)$ and $\Hom_{0,1}(\uA,\uB)$ are 
$\uR$-complete, since the category $\dcpo_{\uR}$ of $\uR$-complete algebras 
is limit closed by Lemma~\ref{compcomp}. \\ 

{\bf Convention.}
For a dcpo $X$, we identify continuous maps $f: X \rightarrow \2$ with 
the corresponding Scott-open subsets $U := f^{-1}(\{1\})\subseteq X$. 
In this way, the dcpo $\2^X$ of continuous functions $u\colon X \to \2$ 
with the pointwise order is identified with the dcpo $\cO(X)$ of Scott-open 
subsets of $X$ ordered by inclusion.

\subsection{Angelic nondeterminism}

The prototype algebra for angelic nondeterminism is $\cH(\Si)$, i.e.\
the algebra $\2_\vee$. We will
see that in this case the observationally-induced powerdomains agree
with the classical Hoare powerdomains.

First we reproduce the standard representation of the Hoare powerdomain
(see, e.g.\ \cite{dom}): 

\begin{prop}
The Hoare powerdomain over a dcpo $X$ is (isomorphic to) the
dcpo $\cH(X)$ of nonempty Scott-closed subsets of $X$ ordered 
by inclusion, the join-semilattice operation being 
binary union; the natural embedding $\eta_X^H\colon X\to \cH(X)$ 
sends every $x\in X$ to its Scott closure $\da x$.
\end{prop}

Indeed, if $f$ is a continuous map from $X$ to some
dcpo-join-semilattice $\uB$, then $\widetilde f(C) = \bigvee_{x\in C} f(x)$ 
is the unique continuous homomorphic extension of $f$ along $\eta^H$. 
From this definition of the extension it follows immediately that $f\leq f'$ 
implies $\widetilde f\leq \widetilde f'$. This implies that the Hoare 
powerdomain is internally free over $X$. 

For our purposes we want to pass to a functional representation:

\begin{lem}\label{lem:funchoare}
The Hoare powerdomain over a dcpo $X$ is (isomorphic to) the
dcpo-join-semilattice $\Hom_{0,1}(\2_\vee^X,\2_\vee)$ of all continuous
semilattice homomorphism from $\2_\vee^X$ to $\2_\vee$ preserving
bottom and top. The embedding $\eta_X^H\colon X \to 
\Hom_{0,1}(\2_\vee^X,\2_\vee)$ is given by evaluation $\eta_X^H(x)(u)=u(x)$.
\end{lem}

The isomorphism from the $\cup$-semilattice of nonempty closed subsets
of $X$ to the function space $\Hom_{0,1}(\2_\vee^X,\2_\vee)$ is given
by assigning to every nonempty 
closed subset $C$ of $X$ the map $\phi\colon \2_\vee^X\to\2_\vee$
defined by $\phi(U)=0$ if and only if $U\cap C= \emptyset$. 
$\Hom_{0,1}(\2_\vee^X,\2_\vee)$ is $\2_\vee$-complete by
Lemma \ref{lem:complete}; it is not only a dcpo-subsemilattice of
$\2_\vee^{\2^X}$ but even the smallest dcpo-subsemilattice of
$\2_\vee^{\2^X}$ containing the point evaluations $\eta_X^H(x), x\in
X$, since $\Hom_{0,1}(\2_\vee^X,\2_\vee)$ is the free
dcpo-join-semilattice over $X$. Hence $\eta_X^H\colon 
 X\to\Hom_{0,1}(\2_\vee^X,\2_\vee)$ is the  $\2_\vee$-repletion of $X$.  
Altogether we have: 

\begin{thm}\label{th:hoare}
The dcpo-join-semilattice $\Hom_{0,1}(\2_\vee^X,\2_\vee)$ is the
$\2_\vee$-repletion of $X$. 
Up to isomorphism, the Hoare powerdomain $\eta_X^H\colon X\to\cH(X)$
over a dcpo $X$ agrees with the $\2_\vee$-repletion of $X$.
\end{thm}

\subsection{Demonic nondeterminism}

The prototype for demonic nondeterminism is $\cS(\Si)$, i.e.\ the
algebra $\2_\wedge$.

There is no explicit description available for the free
dcpo-meet-semilattice over an arbitrary dcpo $X$. However, for
continuous dcpo's $X$ the free dcpo-meet-semilattice over $X$ will 
give rise to the classical Smyth powerdomain $\cS(X)$.

The standard representation of the Smyth powerdomain over a continuous
dcpo is the following (\cite{dom}):

\begin{prop}
The Smyth powerdomain over a continuous dcpo $X$ is (isomorphic to) the
dcpo $\cS(X)$ of nonempty compact saturated subsets $Q$ of $X$ ordered 
by reverse inclusion $\supseteq$, the `meet'-semilattice operation being 
binary union; the natural embedding $\eta_X^S\colon X \to \cS(X)$ sends
every $x\in X$ to its upper closure ${\ua}x$. 
\end{prop}

Indeed, if $f$ is a continuous map from a continuous dcpo $X$ to some
dcpo-meet-semi\-lattice $B$, its unique extension to a meet-semilattice
homomorphism is given by $\widetilde f(Q) = \dsup_{F\in
  \cF}\bigwedge_{x\in F} f(x)$, where $F$ ranges over  the collection
$\cF$ of finite subsets $F$ of $X$ such that $Q$ is contained in the
Scott-interior of ${\ua}F$. From the definition of the extension it
follows directly that, if $f\leq f'$, then $\widetilde f\leq
\widetilde f'$. Thus, the Smyth powerdomain over a continuous dcpo is
internally free.    

For any dcpo $X$, the dcpo-meet-semilattice
$\Hom_{0,1}(\2_\wedge^X,\2_\wedge)$ is isomorphic to the collection
$\mathcal{OF}(X)$ of all Scott-open (proper) filters of the dcpo $\cO(X)$ 
of all Scott-open subsets of $X$ ordered by inclusion with binary
intersection as semilattice operation. The isomorphism is given by
assigning to $\phi\in \Hom_{0,1}(\2_\wedge^X,\2_\wedge)$ the
inverse image $\phi^{-1}(1)$. 

If the dcpo $X$ is sober, $\mathcal{OF}(X)$ is isomorphic to the 
collection of all nonempty compact saturated subsets $Q$ of $X$ ordered 
by reverse inclusion $\supseteq$ with binary union as semilattice operation 
by \cite[Theorem II-1.20]{dom}. The isomorphism is given by assigning to 
every $\cU\in\mathcal{OF}(X)$ its intersection $Q = \bigcap \cU$. Since 
every continuous dcpo is sober, we have the functional representation:

\begin{lem}\label{lem:funcsmyth}
The Smyth powerdomain over a continuous dcpo $X$ is (isomorphic to) the
dcpo-meet-semilattice $\Hom_{0,1}(\2_\wedge^X,\2_\wedge)$ of all continuous
semilattice homomorphism from $\2_\wedge^X$ to $\2_\wedge$ preserving
bottom and top. The embedding $\eta_X^S\colon X\to
\Hom_{0,1}(\2_\wedge^X,\2_\wedge)$ is given by evaluation $\eta_X^S(x)(u)=u(x)$. 
\end{lem}




By Lemma \ref{lem:complete},
$\Hom_{0,1}(\2_\wedge^X,\2_\wedge)$ is $\2_\wedge$-complete. If $X$ is
continuous, then it is the free dcpo-meet-semilattice over $X$ by
\ref{lem:funcsmyth}. Thus, there is no proper dcpo-meet-subsemilattice
containing the evaluation maps $\eta_X^S(x), x\in X$. We conclude:

\begin{thm}\label{th:smyth}
For every continuous dcpo $X$, the dcpo-meet-semilattice
$\Hom_{0,1}(\2_\wedge^X,\2_\wedge)$ is the $\2_\wedge$-repletion of $X$. 
Thus up to isomorphism, the Smyth powerdomain $\eta_X^S \colon X\to\cS(X)$ 
over a continuous dcpo $X$ agrees with the $\2_\wedge$-repletion of $X$.
\end{thm}

\subsection{Erratic nondeterminism} \label{plotpower}

The prototype algebra for erratic nondeterminism is $\mathcal{P}(\Si)$,
i.e.\ the algebra whose underlying dcpo is the three element chain $\A$ 
with the semilattice operation $\funion$ as defined at the 
beginning of this section.


Erratic nondeterminism seems to be much more complicated to handle 
than angelic or demonic nondeterminism.
In contrast to the angelic and demonic case we arrive at the conclusion that
the Plotkin powerdomain does not agree with $\A$-repletion; it is
in fact properly contained in the repletion. The following theorem 
is expressing this claim. In this theorem we need the additional hypothesis 
that the dcpo is compact for its Scott topology. This is equivalent to the 
requirement that there is a finite subset $F$ such that every element is 
above some element in $F$. In particular, if the dcpo has a bottom element, 
then it is Scott-compact.

\begin{thm}\label{th:main}
(a) For a compact continuous dcpo $X$, the dcpo-semilattice
$\Hom_{0,1}(\A^X,\A)$ of all continuous semilattice
homomorphisms $\phi\colon \A^X\to\A$ preserving bottom and top is the
$\A$-repletion of $X$, that is, $\Hom_{0.1}(\A^X,\A)$ is the least
$\A$-complete dcpo-subsemilattice of $\A^{\A^X}$ containing all the
evaluation maps $\eta_X^P(x)=(u\mapsto u(x)), x\in X$.

(b) For a continuous dcpo $X$, the Plotkin power domain $\cP (X)$ is
(isomorphic to) the least dcpo-subsemilattice of $\A^{\A^X}$ containing all the
evaluation maps $\eta_X^P(x)=(u\mapsto u(x)), x\in X$.
\end{thm}

 Let us indicate, how big the gap between the Plotkin powerdomain and
the repletion $\Hom_{0.1}(\A^X,\A)$ really is: As we will see in
\ref{sec:lens}, the elements of the repletion of $X$ correspond to 
`formal lenses', that is, pairs $(C,Q)$ of a closed set $C$ and a compact 
saturated set $Q$ in $X$ with nonempty intersection. Each of this formal 
lenses determines a  `real lens' $L = C\cap Q$. The same real lens can be 
represented by many different formal lenses. If the dcpo is continuous and 
coherent, the real lenses correspond to the elements of the
Plotkin powerdomain as we recall later in \ref{sec:lens}.

A first step towards a proof of Theorem \ref{th:main} is a consequence
of Lemma \ref{lem:complete}: 
 
\begin{prop}\label{prop:plotkincomplete}
For any dcpo $X$, the dcpo-semilattice $\Hom_{0,1}(\A^X,\A)$ of all
top and bottom preserving semilattice homomorphisms $\phi\colon
\A^X\to\A$ is $\A$-complete.  
\end{prop}

For $\eta_X^P \colon X \to \Hom_{0,1}(\A^X,\A)$ being the $\A$-repletion 
of $X$, it suffices to show now that $\eta^P_X$ is $\A$-equable, i.e.\
that for every continuous $f : X \to \A$ there exists a unique 
Scott-continuous homomorphism $\wh f : \Hom_{0,1}(\A^X,\A) \to \A$  
with $\wh f \circ \eta^P_X = f$ and $f\mapsto \wh f$ is continuous. 
Actually, existence and continuity of $\wh f$ is not the problem since 
one may take $\wh f(\varphi) = \varphi(f)$. The question rather is uniqueness.
For answering this question positively we prove an auxiliary result.

\subsubsection{An auxiliary observation}\label{auxobs}

\begin{enumerate}
\item Let $A_1$ and $B_1$ be dcpo-join-semilattices and $A_2$ and $B_2$
dcpo-meet-semilattices. 
\end{enumerate}
The direct product $A_1\times A_2$ (and similarly $B_1\times B_2$) is a 
dcpo-semilattice with the product order $$(a_1,a_2)\leq (a_1',a_2')
\iff a_1\leq a_1',a_2\leq a_2'$$ and the semilattice operation 
$$(a_1,a_2)\funion (a_1',a_2') = (a_1\vee a_1',a_2\wedge a_2').$$  
We denote by $\pi_i\colon A_1\times A_2\to A_i,\ i=1,2,$ the canonical
projections onto the factors. We use the same notation for the
projections of $B_1\times B_2$ onto its factors. 
\begin{enumerate}[resume]
\item We suppose in addition that $B_1$ has a top element $1$ and
  $B_2$ a bottom element $0$.
\end{enumerate}
 Then we have embeddings $\epsilon_i\colon B_i\to
B_1\times B_2$ defined by $\epsilon_1(b_1)=(b_1,0)$ and
$\epsilon_2(b_2)=(1,b_2)$, respectively, which are semilattice
homomorphisms. 
\begin{enumerate}[resume]
\item[(3)] Suppose that $B$ is a dcpo-subsemilattice of
$B_1\times B_2$ containing $(b_1,0)$ and $(1,b_2)$ for all $b_1\in
B_1, b_2\in B_2$. 
\end{enumerate}
\begin{diagram} 
            &             &B_1\times B_2&             &\\
            &\ruTo^{\epsilon_1}&\uInto         &\luTo^{\epsilon_2}&\\
B_1         &\pile{\lTo^{\pi_1}\\
  \rTo_{\epsilon_1}}&B&\pile{\rTo^{\pi_2}\\ 
\lTo_{\epsilon_2}}&B_2\\
\dTo<{\Phi_1}&            &\dTo<{\Phi}  &            &\dTo>{\Phi _2}\\
A_1         &\lTo_{\pi_1} &A_1\times A_2&\rTo_{\pi_2} &A_2
\end{diagram}
\begin{obs}\label{obs}
For a continuous semilattice homomorphism $\Phi\colon B\to A_1\times
A_2$ consider the continuous semilattice homomorphisms $\Phi_i=
\pi_i\Phi\epsilon_i\colon B_i\to A_i, i=1,2,$ that is, $\Phi_1(b_1)
=  \pi_1(\Phi(b_1,0))$ and  $\Phi_2(b_2) =\pi_2(\Phi(1,b_2))$ for all
$b_1\in B_1, b_2\in B_2$. 
Then the above diagram commutes, that is, $\Phi(b_1,b_2) =
(\Phi_1(b_1), \Phi_2(b_2))$ for all $b=(b_1,b_2)\in B$.   
\end{obs}

\begin{proof}
Let $\Phi\colon B\to A_1\times A_2$ be a semilattice
homomorphism. The composed maps $\psi_i=\pi_i\Phi\colon B \to A_i$ are semilattice homomorphisms, too. We claim that $\psi_1$
does not depend on the second and $\psi_2$ not on the first argument.

Indeed, for $(b_1,b_2), (b_1,b_2')\in B$ we have:
$\psi_1(b_1,b_2),\psi_1(b_1,b_2') \leq
\psi_1(b_1,b_2)\vee\psi_1(b_1,b_2') =
\psi_1((b_1,b_2)\funion(b_1,b_2')) = 
\psi_1(b_1,b_2\wedge b_2') \leq \psi_1(b_1,b_2),\psi_1(b_1,b_2')$. Thus 
 $\psi_1(b_1,b_2)=\psi_1(b_1,b_2')$. The order
 dual argument yields the analogous statement for $\psi_2$. 
 
It follows that the maps $\Phi_i\colon B_i\colon\to A_i, i=1,2,$
defined by $\Phi_1(b_1)=\psi_1(b_1,0)=\pi_1\Phi(b_1,0)$ and
$\Phi_2(b_2)=\psi_2(1,b_2)=\pi_2\Phi(1,b_2)$ have the desired properties.    
\end{proof}

 \subsubsection{Formal lenses and real lenses}\label{sec:lens}

Our aim in this subsection is to develop a more conceptual understanding 
of the elements of $\Hom_{0,1}(\A^X,\A)$ as so-called \emph{formal lenses} 
in a dcpo $X$.  

Recall that $\A$ can be seen as the subsemilattice of $\2_\vee\times
\2_\wedge$ consisting of those elements $(a_1,a_2)$ with $a_1\geq a_2$.

Consequently, every $f \in \A^X$ can be written uniquely as a pair
$(u_1,u_2)\in \2_\vee^X\times\2_\wedge^X$ such that $u_1\geq u_2$ or,
alternatively, as a pair $(U_1,U_2)$ of open subsets of $X$ with
$U_1\supseteq U_2$.
 
\begin{lem}\label{lem:3} 
$\Hom_{0,1}(\A^X,\A)$ agrees with the dcpo-subsemilattice of all pairs
$(\varphi_1,\varphi_2)$ in $\Hom_{0,1}(\2_\vee^X,\2_\vee)\times
\Hom_{0,1}(\2_\wedge^X,\2_\wedge)$ such that $\phi_1\geq\phi_2$.   
\end{lem}

\begin{proof}
We want to apply Observation \ref{obs}:  $\A^X$ is
the dcpo-subsemilattice of $\2_\vee^X\times \2_\wedge^X$ of all
$u=(u_1,u_2)$ such that $u_1\geq u_2$. In particular $(u_1,0)\in
\A^X$ and $(1,u_2)\in \A^X$ for all $u_1,u_2\in\2^X$, where $0$ and
$1$ denote also the constant functions with value $0$ and $1$
respectively. Thus the 
hypotheses of $\ref{obs}$ are satisfied. We conclude that the
continuous semilattice homomorphisms $\phi\colon\A^X\to
\2_\vee\times\2_\wedge$ 
are those maps for which there are continuous semilattice homomorphisms
$\phi_1\colon \2_\vee^X\to\2_\vee$ and $\phi_2\colon
\2_\wedge^X\to\2_\wedge$ such that
$\phi(u_1,u_2)=(\phi_1(u_1),\phi_2(u_2))$. The maps $\phi$ and
$\phi_1,\phi_2$ are related by the formulas
$\phi_1(u_1)=\pi_1\phi(u_1,0)$ and
$\phi_2(u_2)=\pi_2\phi(1,u_2)$. Clearly, $\phi$ preserves 
bottom and top if and only if both $\phi_1$ and $\phi_2$ preserve
bottom and top. 
Finally, $\phi$ maps $\A^X$ into $\A$ if and only if $\phi_1\geq\phi_2$.  
\end{proof}

According to Lemma \ref{lem:funchoare} and Lemma \ref{lem:funcsmyth}, the
elements $\phi_1$ of $\Hom_{0,1}(\2_\vee^X,\2_\vee)$ correspond to nonempty
closed subsets $C$ of $X$, namely $C=X\setminus \bigcup\phi_1^{-1}(0)$,
and the elements $\phi_2$ of $\Hom_{0,1}(\2_\wedge^X,\2_\wedge)$
correspond to Scott-open filters $\cU$ of $\cO(X)$, namely $\cU =
\phi_2^{-1}(1)$. The condition $\phi_1\geq \phi_2$ corresponds to the
requirement $C\cap U\neq \emptyset$ for all $U\in \cU$. Hence
$\Hom_{0,1}(\A^X,\A)$ is the functional representation of the
dcpo-semilattice of formal lenses according to the following
definition:

\begin{defi}   {\rm
A \emph{formal lens} of a dcpo $X$ consists of a pair $(C,\cU)$ of a
nonempty closed subset $C$ of $X$ and a Scott-open filter $\cU$ of
$\cO(X)$  such that $C\cap U\neq \emptyset$ for all $U\in\cU$. }  
\end{defi}

The formal lenses form a dcpo-subsemilattice of $\Cc(X)\times
\mathcal{OFilt}(X)$, where $\Cc(X)$ is the dcpo-$\cup$-semilattice  of
nonempty closed subsets of $X$ and
$\mathcal{OFilt}(X)$  is the dcpo-$\cap$-semilattice of Scott-open
filters of $\cO(X)$.

In the case where the dcpo $X$ is sober w.r.t.\ its Scott topology,
under assumption of AC the Scott-open filters $\cU$ can be replaced by their
intersections $Q=\bigcap \cU$ which are the nonempty Scott-compact saturated
subsets of $X$. Thus the formal lenses of a sober dcpo $X$ are the
pairs $(C,Q)$ of closed subsets $C$ and compact saturated subsets $Q$
of $X$ such that $C\cap Q\neq\emptyset$. But remember that the dcpo of
nonempty compact saturated sets is ordered by reverse inclusion
$\supseteq$ and the meet-operation is $\cup$.

Notice that formal lenses $(C,Q)$ are not determined by the intersection 
$L := C \cap Q$ which can be considered as the {\bf real lens}
corresponding to the formal lens $(C,Q)$. This real lens $L$ can be
represented by the ``normalized'' formal lens $(C_L,Q_L)$ where $C_L$
is the closure of $L$ for the Scott topology and $Q_L= {\ua}L$.

In most cases the dcpo of formal lenses is much bigger than the dcpo of real lenses.
As an example consider the two element set $X=\{0,1\}$ with the discrete order. 
There are three real lenses, the nonempty subsets. But there are four more formal 
lenses. The real lens $\{1\}$ has three representations by formal lenses, namely 
$\{1\} = \{1\}\cap \{1\} = \{1\}\cap \{0,1\} = \{0,1\} \cap \{1\}$,
and so has $\{0\}$. Thus, for flat dcpos $X$ with more than one element 
$\A(X)$ is different from the Plotkin powerdomain $\mathcal{P}(X)$ and the
same applies to their liftings $X_\bot$.

\subsubsection{The $\A$-repletion of a compact continuous dcpo}

Since $\Hom_{0,1}(\A^X,\A)$ is $\A$-complete by Proposition
\ref{prop:plotkincomplete},
the question now is whether the map $\eta^P \colon X \to
\Hom_{0,1}(\A^X,\A)\colon x \mapsto \lambda \phi. \phi(x)$ is $\A$-equable.
I.~Battenfeld in his 2013 paper has shown that this
is actually true for $X = 2$, the discrete poset. In this section
we give a positive answer for compact continuous dcpos $X$. 

As before, we consider $\A$ to be a dcpo-subsemilattice of
$\2_\vee\times\2_\wedge$. We write
$\pi_1\colon\2_\vee\times\2_\wedge\to\2_\vee$ and
$\pi_2\colon\2_\vee\times\2_\wedge\to\2_\wedge$ for the respective canonical
projections.

By Lemma \ref{lem:3}, $\Hom_{0,1}(\A^X,\A)$ agrees with the
dcpo-subsemilattice of all 
$(\varphi_1,\varphi_2)\in\Hom_{0,1}(\2_\vee^X,\2_\vee)\times
\Hom_{0,1}(\2_\wedge^X,\2_\wedge)$ such that $\phi_1\geq\phi_2$. 
We write $\pi_H$ and $\pi_S$ for the projections from
$\Hom_{0,1}(\A^X,\A)$ to   $\Hom_{0,1}(\2_\vee^X,\2_\vee)$ and
$\Hom_{0,1}(\2_\wedge^X,\2_\wedge)$, respectively.
The maps $\eta^P \colon X \to  \Hom_{0,1}(\A^X,\A)$, 
$\eta^H \colon X \to \Hom_{0,1}(\2_\vee^X,\2_\vee)$ and $\eta^S \colon X
\to\Hom_{0,1}(\2_\wedge^X,\2_\wedge)$ are all three
given by $\lambda x.\lambda u. u(x)$ as usual and one easily checks that
the diagram
\begin{diagram}
& & X & & \\
& \ldTo^{\eta^H} & \dTo_{\eta^P} & \rdTo^{\eta^S} & \\
\Hom_{0,1}(\2_\vee^X,\2_\vee) & \lTo_{\pi_H} & \Hom_{0,1}(\A^X,\A) & \rTo_{\pi_S} & \Hom_{0,1}(\2_\wedge^X,\2_\wedge) 
\end{diagram}
commutes. In order to apply Observation \ref{obs}, we need the
hypothesis that $\Hom_{0,1}(\2_\wedge^X,\2_\wedge)$ has a least and
$\Hom_{0,1}(\2_\vee^X,\2_\vee)$ a greatest element. While
$\Hom_{0,1}(\2_\vee^X,\2_\vee)$ always has a greatest element --- the map
$\phi_\top$ that maps all open sets to top except for the empty set,
we have to suppose $X$ to be compact for the Scott topology so that
$\Hom_{0,1}(\2_\wedge^X,\2_\wedge)$ has a least element. Indeed,
the map $\phi_\bot$ mapping all open sets to bottom except for the
whole space is a meet-semilattice homomorphism, and it is continuous
if and only if $X$ is compact. Thus, we assume $X$ to be Scott-compact.
We now have continuous maps
$\epsilon_H\colon\Hom_{0,1}(\2_\vee^X,\2_\vee)\to\Hom_{0,1}(\A^X,\A)\colon
\phi_1\mapsto (\phi_1,\phi_\bot)$ with $\pi_H\circ \epsilon_H = \id$,
and
$\epsilon_S\colon\Hom_{0,1}(\2_\wedge^X,\2_\wedge)\to\Hom_{0,1}(\A^X,\A)\colon   
\phi_2\mapsto (\phi_\top,\phi_2)$ with $\pi_S\circ \epsilon_S
= \id$. 

Now we can apply the decomposition Observation \ref{obs} and we obtain:

\begin{lem}\label{hunique}
Let $X$ be a compact dcpo and 
$\Phi\colon\Hom_{0,1}(\A^X,\A)\to \A$ a continuous semilattice
homomorphism. Then there exist unique continuous
semilattice homomorphisms $\Phi_H\colon \Hom_{0,1}(\2_\vee^X,\2_\vee) \to
\2_\vee$ and $\Phi_S\colon\Hom_{0,1}(\2_\wedge^X,\2_\wedge) \to\2_\wedge$ 
such that
the following diagram commutes:
\begin{diagram}
\Hom_{0,1}(\2_\vee^X,\2_\vee) & \pile{\lTo^{\pi_H}\\
  \rTo_{\epsilon_H}} & \Hom_{0,1}(\A^X,\A) & 
\pile{\rTo^{\pi_S}\\ \lTo_{\epsilon_S}}& \Hom_{0,1}(\2_\wedge^X,\2_\wedge) \\  
\dTo^{\Phi_H} & & \dTo_\Phi & & \dTo_{\Phi_S} \\
\2_\vee & \lTo_{\pi_1} & \A & \rTo_{\pi_2} & \2_\wedge \\
\end{diagram}
\end{lem}

Now we are ready to show the result that we were aiming for.

\begin{prop}\label{mainth}
For compact continuous dcpos $X$ the map $\eta_X^P\colon X
\to\Hom_{0,1}(\A^X,\A)$ is $\A$-equable. 
\end{prop}
\begin{proof}
Given $u \colon X \to \A$, a continuous homomorphic extension 
$\Phi\colon\Hom_{0,1}(\A^X,\A)\to \A $ of $u$ along 
$\eta_X^P$ is given by $\Phi(\varphi) = \varphi(u)$. 
Clearly, $\Phi$ is a continuous semilattice homomorphism and
we have $\Phi(\eta_X^P(x)) = \eta_X^P(x)(u) = u(x)$  for all $x \in X$,
that is, $\Phi$ extends $u$ along $\eta_X^P$. The extension operator
$u\mapsto \Phi$ is continuous since it is $\lambda$-definable as
$\lambda u.\lambda \phi.\ \varphi(u)$.

From Lemma~\ref{hunique} it follows that $\Phi$ is uniquely determined 
by the semilattice homomorphisms $\Phi_H$ and $\Phi_S$. We have
$\Phi_H \circ \eta_X^H = \pi_1 \circ u$ and $\Phi_S \circ \eta_X^S =
\pi_2 \circ u$. Since $\eta_X^H : X \to \cH(X)$ is internally free
it is also $\2_\vee$-equable. Since $\eta_X^S :  X \to \cS(X)$ is 
internally free it is also $\2_\wedge$-equable. Thus, the maps $\Phi_H$ 
and $\Phi_S$ are uniquely determined in a continuous way by $\pi_1 \circ u$ 
and $\pi_2 \circ u$, respectively. Thus, the map $\Phi$ is uniquely 
determined by $u$ in a continuous way.
\end{proof}

The previous proposition together with Proposition \ref{prop:plotkincomplete}
finishes the proof of part (a) of our main theorem \ref{th:main}.

\subsubsection{The Plotkin powerdomain and $\A$-valuations}

We proceed to a proof of assertion (b) in our Main Theorem~\ref{th:main}. 
We could give a direct proof. But we prefer to use R.~Heckmann's work in
\cite{heckabstvals} on $\A$-valuations and the Plotkin powerdomain. 
 
Heckmann has called a continuous map  $\alpha\colon\2^X\to\A$ an
\emph{$\A$-valuation} on the dcpo $X$, if it preserves bottom and top and,
moreover, satisfies the following two conditions:
\begin{enumerate}[label=(H\arabic*)]
\item if $\alpha(U) =\bot$ then $\alpha(U\cup V)=\alpha(V)$,
\item if $\alpha(U) =\top$ then $\alpha(U\cap V)=\alpha(V)$.
\end{enumerate} 
The collection $\HAVal(X)$ of Heckmann's $\A$-valuations is a
dcpo-subsemilattice of $\A^{\2^X}$.

We want to adapt Heckmann's $\A$-valuations to our setting.
We can identify the semilattice $\HAVal$ of Heckmann's $\A$-valuations 
with a dcpo-subsemilattice of $\Hom_{0,1}(\A^X,\A)$. For this purpose
we consider $\A$ as a join- and as a meet-semilattice with the operations
$$a\vee b =\max(a,b), \ \ \ a\wedge b = \min(a,b)$$ 
where $\max$ and $\min$ refer do the dcpo-ordering of $\A$.

\begin{defi}
A continuous map $\phi\colon \A^X\to\A$ will be called an
\emph{$\A$-valuation} on the dcpo $X$, if it preserves $0$ and $1$
and satisfies the conditions
\begin{enumerate}[label=(H\arabic*)]
\item if $\phi(u) =\bot$ then $\phi(u\vee v)=\phi(v)$,
\item if $\phi(u) =\top$ then $\phi(u\wedge v)=\phi(v)$.
\end{enumerate} 
We denote by $\AVal(X)$ the collection of these $\A$-valuations.
\end{defi} 

An easy calculation shows that one can pass from Heckmann's
$\A$-valuations $\alpha$ to our 
$\A$-valuations by defining $\ov\alpha(U_1,U_2)
=\alpha(U_1)\funion\alpha(U_2)$. Conversely, from an $\A$-valuation
$\phi$ in our sense on obtains an $\A$-valuation in the sense of
Heckmann by defining $\ov\phi(U) = \phi(U,U)$. (Here we have used that the
elements of $\A^X$ can be represented as pairs of open sets
$u=(U_1,U_2)$ with $U_1\supseteq U_2$.)  

For working with $\A$-valuations in our sense, another
characterization is useful:
As before, we consider
$\A$ to be a subsemilattice of $\2_\vee \times \2_\wedge$ and we
denote by $\pi_i, i=1,2,$ the canonical projections onto the two
factors. Similarly, we represent $\Hom_{0,1}(\A^X,\A)$ as the
dcpo-subsemilattice of all $\phi=(\phi_1,\phi_2)\in
\Hom_{0,1}(\2_\vee^X,\2_\vee) \times
\Hom_{0,1}(\2_\wedge^X,\2_\wedge)$ such that $\phi_1\geq\phi_2$.

\begin{lem}
Let $X$ be any dcpo. 
\begin{enumerate}[label=\({\alph*}]
\item For every Heckmann $\A$-valuation $\alpha\colon\2^X\to\A$,
let $\ov\alpha_i=\pi_i\circ\alpha\colon\2^X\to\2, i=1,2,$ and 
$\ov\alpha =(\ov\alpha_1,\ov\alpha_2)$. Then
$\ov\alpha\in\Hom_{0,1}(\A^X,\A)$. 

\item A $\phi=(\phi_1,\phi_2)\in\Hom_{0,1}(\A^X,\A)$ is of the form
$\ov\alpha$ for some Heckmann $\A$-valuation $\alpha$ if and only if 
it satisfies the following two conditions:
\begin{enumerate}[label=(H\arabic*')]
\item if $\varphi_1(U) = \bot$ then $\varphi_2(V) =
  \varphi_2(U \cup V)$  
\item if $\varphi_2(U) = \top$ then $\varphi_1(V) =
  \varphi_1(U \cap V)$  
\end{enumerate} 

\item The elements $\phi=(\phi_1,\phi_2)\in\Hom_{0,1}(\A^X,\A)$
  satisfying (H1') and (H2') are exactly the $\A$-valuations according
  to the definition above. They form a dcpo-subsemilattice of
  $\Hom_{0,1}(\A^X,\A)$ and $\alpha\mapsto\ov\alpha$ is a
  dcpo-semilattice isomorphism of $\HAVal$ onto $\AVal(X)$.
\end{enumerate}
\end{lem}

\begin{proof}
(a) It is clear that the $\alpha_i$ are continuous and that they
preserve bottom and top. We check that $\alpha_1$ is a join-semilattice
homomorphism. It suffices to show: if $\alpha_1(U) =\alpha_1(V) = 0$ then
$\alpha_1(U\cup V)=0$, and this is a direct consequence of Heckmann's
condition (H1). One uses condition (H2) in a similar way to show that
$\alpha_2$ is a meet-semilattice homomorphism. Since $\pi_1\geq \pi_2$,
we also have $\alpha_1=\pi_1\circ\alpha\geq\pi_2\circ\alpha\geq\alpha_2$. 

(b) For every Heckmann $\A$-valuation $\alpha$, $\alpha_1$ and
$\alpha_2$ satisfy the 
conditions (H1') and (H2'). Indeed, if $\alpha_1(U)=\pi_1(\alpha(U)) =
0$, then $\alpha(U)=0$, whence $\alpha(U\cup V)=\alpha(V)$ by
condition (H1); thus, $\alpha_2(U\cup V)=\pi_2\alpha(U\cup
V)=\pi_2\alpha(U) = \alpha_2(V)$, which proves condition (H1').
Similarly one proceeds for (H2').

Conversely, for every $\phi=(\phi_1,\phi_2)\in\Hom_{0,1}(\A^X,\A)$ we
define a map $\ov\phi\colon \2^X\to\A$ by
$\ov\phi(U)=(\phi_1(U),\phi_2(U))$. Clearly $\ov\phi$ is
continuous and preserves bottom and top. If $\phi$ satisfies
(H1') and (H2'), then $\ov\phi$ is a Heckmann $\A$-valuation. Indeed, suppose
$\ov\phi(U)=\bot$. Then $(\phi_1(U),\phi_2(U))=\bot =(0,0)$. Thus,
firstly, $\phi_1(U) =0$ which implies that $\phi_1(U\cup
V)=\phi_1(U)\vee\phi_1(V) =\phi_1(V)$, where we have used that
$\phi_1$ is a join-homomorphism. Secondly,  $\phi_1(U)=0$ also
implies $\phi_2(U\cup V)= \phi_2(V)$ by condition (H1'). Thus
$\ov\phi(U\cup V)=(\phi_1(U\cup V),\phi_2(U\cup V))=
(\phi_1(V),\phi_2(V)) = \ov\phi(V)$. Thus $\ov\phi$ satisfies
(H1). Similarly,  $\ov\phi$ satisfies (H2).

(c) Clearly $\ov{\ov\alpha}=\alpha$ and $\ov{\ov\phi}=\phi$. Thus,
$\alpha\mapsto\ov\alpha$ and $\phi\mapsto\ov\phi$ are mutually inverse
bijections. These bijections are order preserving by their very
definition. Hence they are dcpo-isomorphisms.
It is immediate from the definition that $\alpha\mapsto\ov\alpha$ is a
semilattice homomorphism from $\AVal(X)$ to
$\Hom_{0,1}(\A^X,\A)$. 
\end{proof}

 

There is a natural map $\delta \colon X\to\HAVal(X)$ which to every 
$x\in X$ assigns the map $\delta(x)\colon \2^X\to \A$ defined 
as follows:
\[
\delta(x)(U)=\begin{cases} \top & \text{ if } x\in U,\\
                           \bot & \text{ else }.
             \end{cases}
\]
Clearly, $\delta (x)$ is a Heckmann $\A$-valuation and $\delta$ depends
continuously on $x$. For every nonempty finite subset $F$ of $X$, we
can form the $\A$-valuation $\delta(F) =\ \bigfunion _{x\in F}\delta(x)$ 
which can be defined by
\[
\delta(F)(U) = \begin{cases} \top & \text{ if } F\subseteq U,\\
                             \bot & \text{ if } F\cap U = \emptyset,\\
                             m    & \text{ else}.
               \end{cases}
\] 
Moreover, the following diagram commutes:
\begin{diagram}
        &&X&&&\\
        &\ldTo^\delta&&\rdTo^{\eta^P}&&\\
\HAVal(X)&&\rTo_{\alpha\mapsto\ov\alpha}&& \AVal(X) & \;\subseteq\Hom_{0,1}(\A^X,\A) 
\end{diagram} 

For continuous dcpos Heckmann has shown \cite[Theorem
6.1]{heckabstvals} that $\HAVal(X)$ is (isomorphic to) the Plotkin
powerdomain over $X$ and that the $\delta(F)$ for nonempty finite subsets $F$
of $X$ form a basis. Using the preceding commuting diagram we have: 

\begin{thm}\label{th:heckmann}
For any continuous dcpo $X$, the dcpo-semilattice $\AVal(X)$ of
$\A$-valuations is 
(isomorphic to) the Plotkin powerdomain $\cP(X)$, the canonical map
being $\eta^P\colon X\to\AVal(X)$. Moreover,
$\AVal(X)$ is a continuous dcpo. The $\A$-valuations $\eta^P(F)$
form a basis, when $F$ ranges over the nonempty finite subsets of $X$.
\end{thm}

We now can finish the proof of part (b) of Theorem \ref{th:main}:

\begin{prop}
For a continuous dcpo $X$, the Plotkin powerdomain is (isomorphic to) the least
dcpo-subsemilattice of $\A^{\A^X}$ containing all the point evaluations
$\eta_X^P(x) = \lambda u.\ u(x)$, $x\in X$.
\end{prop}

\begin{proof}
From Heckmann's Theorem \ref{th:heckmann} we know that $\HAVal(X)$ is
the free dcpo-semilattice over $X$. Thus $\HAVal(X)$ has no proper
dcpo-subsemilattice containing all the $\delta(x), x\in X$. By the
isomorphism established in the preceding lemma, there is a
dcpo-semilattice isomorphism from $\HAVal(X)$ onto the
dcpo-subsemilattice $\AVal(X)$ of $\A^{\A^X}$ mapping $\delta(x)$ to
the evaluation map $\eta_X^P(x)$, $x\in X$. Thus this subsemilattice is
the least dcpo-subsemilattice of $\A^{\A^X}$ containing all the
$\eta_X^P(x)$, $x\in X$. 
\end{proof}

\begin{prob}{\rm
In this context it seems appropriate to recall Birkhoff's theorem
\cite[pages 143f.]{Bir} 
from universal algebra (in the category of sets): For a
finitary signature $\Omega$, consider an $\Omega$-algebra $\uR$. For any 
set $X$ let $\eta\colon X \to \uR^{\uR^X}$ denote the canonical map defined 
by $\eta(x)(f) = f(x)$ for $x\in X$ and $f\colon X\to R$. Then the 
$\Omega$-subalgebra $F(X)$ of the product algebra $\uR^{\uR^X}$ generated by 
the `projections' $\eta(x)$, $x\in X,$  is free  in the equational class 
of $\Omega$-algebras satisfying all the equational laws that hold in $\uR$. 

One would like to transpose this result into the context of
$\dcpo_\Omega$-algebras. Our last proposition seems to indicate that
this is not an hopeless effort. Instead of equational laws one will
also have to consider inequational laws, too. Thus one can ask the
question: \emph{Given a $\dcpo_\Omega$-algebra $\uR$ and a dcpo $X$. 
Is the least $\dcpo_\Omega$-subalgebra of $\uR^{R^X}$ containing all the
evaluation maps $\eta(x)=\lambda u.\ u(x)$ free in the class of
$\dcpo_\Omega$-algebras satisfying all equational and inequational
laws that hold in $\uR$?} }
\end{prob} 

\begin{rem}{\rm
We have seen that, in general, $\A$-valuations are more general than
lenses. In order to describe $\A$-valuations in general,
J. Goubault-Larrecq \cite{jglduality} has 
introduced the notion of a \emph{quasilens}. This is a formal lens
$(C,\cU)$ where $C$ is a nonempty closed set and $\cU$ an open filter of
open sets such that $C\subseteq \cl(C\cap U)$ for all $U\in\cU$. 
In sober dcpos, formal lenses are given by pairs $(C,Q)$ where $C$ is
a closed and $Q$ a compact saturated set with nonempty intersection. Such a 
formal lens is a quasilens if $C\subseteq \cl(C\cap U)$ for all open sets $U$
containing $Q$. }
\end{rem}

\subsubsection{The Heckmann conditions and real lenses}

The Heckmann conditions look amazing. But they arise in
a natural way from the point of view of real lenses in a dcpo $X$. 

Let us fix some notation for this section. Firstly suppose that $X$ is
a sober dcpo. Recall that a formal lens is a pair $(C,Q)$ consisting
of a closed subset $C$ and a compact saturated subset $Q$ such that
$C \cap Q\neq \emptyset$. We denote by $V$ the open set $X\setminus C$
and by $L=C\cap Q$ the `real' lens associated with the formal lens.
  
The formal lenses are in a one-to-one correspondence with  the
$\phi \in \Hom_{0,1}(\A^X,\A)$. If we represent $\phi$ as a
pair $\phi=(\phi_1,\phi_2)\in \Hom_{0,1}(\2_\vee^X,\2_\vee)\times 
\Hom_{0,1}(\2_\wedge^X,\2_\wedge)$ such that $\phi_1\geq\phi_2$
according to Lemma \ref{lem:3}, then the corresponding formal lens $(C,Q)$ is
obtained  as $C=X\setminus V$, where $V=\bigcup\{U\mid \phi_1(U)=0\}$
is the greatest open set with $\phi_1(V)=0$, and $Q=
\bigcap\{U\mid\phi_2(U)=1\}$. Let $L$ be the real lens $L=C\cap Q$. 

A real lens $L$ has many representations by formal lenses. Among them 
there is a `least' one, namely $(\cl(L),{\ua}L)$. We would like to characterize 
those $\phi\in \Hom_{0,1}(\A^X,\A)$ that correspond to these minimal
representations. This leads to the Heckmann conditions.

\begin{lem}\label{lem:lens}
Given a lens $L$ in a dcpo $X$, define  
maps $\varphi_1$ and $\varphi_2$ from $\2^X$ to $\2$ by
\[ \varphi_1(U) = 0 \;\;\mbox{ iff }\;\; L\cap U = \emptyset \qquad\mbox{ 
  and }\qquad \varphi_2(U) = 1 \;\;\mbox{ iff }\;\; L\subseteq U, \]
then $\phi_L=(\phi_1,\phi_2)$ is an $\A$-valuation
. Moreover $\phi_L$ corresponds to
the minimal formal lens $(\cl(L),\ua L)$ representing $L$.
\end{lem}

\begin{proof}
One easily checks that $\varphi_1$ and $\varphi_2$ are
join- and meet-semilattice homomorphisms,
respectively, preserving bottom and top. For continuity we have
to use the compactness of $L$. Moreover the Heckmann conditions 
are satisfied.  (H1') holds since if $U \cap L = \emptyset$
then $L \subseteq W$ iff $L \subseteq U \cup W$ and (2) holds since if
$L \subseteq U$ then $W \cap L = \emptyset$ iff $W \cap U \cap L =
\emptyset$.  Each of the conditions (H1') and (H2') implies that
$\phi_1\geq\phi_2$. Indeed $\phi_2(U)=1$, by (H2'), implies
$\phi_1(U)=\phi_1(U\cap X) =\phi_1(X)=1$. Thus,
$\phi_L=(\phi_1,\phi_2)$ is an  element of $\Hom_{0,1}(\A^X,\A)$.

Since $X\setminus\cl(L)$ is the greatest open set disjoint from
$L$, the definition  of $\phi_1$ yields
$\bigcup\{U\mid \phi_1(U)=0\}=X\setminus \cl(L)$ and, since the intersection 
of the open neighborhoods of a set is its saturation, we have
$\bigcap\{U \mid \phi_2(U)=1\} = {\ua}L$.
\end{proof}

Of course, one will ask now, whether lenses correspond bijectively
to the $\phi = (\phi_1,\phi_2)\in \Hom_{0,1}(\A^X,\A)$ satisfying the
Heckmann conditions. We have seen in \ref{sec:lens} that this is not so, in
general, even for algebraic dcpos.  Let us try to find sufficient
conditions for this to hold.

\begin{lem}\label{lem:H1'}
Let $X$ be a sober dcpo and suppose that $\phi =(\phi_1,\phi_2) \in
\Hom_{0,1}(\A^X,\A)$ satisfies the Heckmann condition (H1'). For the
formal lens $(C,Q)$ associated with $\phi$ 
we have $Q = \ua(C\cap Q) = {\ua}L$.
\end{lem}

\begin{proof}
We have to show that $Q=\bigcap\{U\mid \phi_2(U)=1\}=\ua L$. For this, it 
suffices to show that if $ L =Q\cap C\subseteq U$, then $Q\subseteq
U$.  Since $\phi_1(V)=0$, condition (H1') tells
us that $\phi_2(U)=\phi_2(U\cup V)$. Now $Q\cap C\subseteq U$ implies
$Q\subseteq U\cup V$. Thus $\phi_2(U\cup V)= 1$, whence $\phi_2(U)=1$,
which implies $Q\subseteq U$ as desired. 
\end{proof}

Starting with a $\phi =(\phi_1,\phi_2)$ satisfying
the Heckmann condition (H2'), we would like to show that
$C=\cl(L)$. For this, it suffices to show that if $W\cap L=W\cap C\cap
Q=\emptyset$ then  $W\cap C=\emptyset$. 

If  we can find an open set $U$ containing $Q$ such that
$W\cap U\subseteq V$, 
then we are on the safe side. Then indeed
$\phi_2(U)=1$. Using (H2') we then have $\phi_1(W\cap
U)=\phi_1(W)$. Since $W\cap U\subseteq V$, we have $\phi_1(W\cap U)=0$
and hence, $\phi_1(W)=0$, whence $W\subseteq V$.        

But to find an open neighborhood $U$ of $Q$ such that $U\cap W\subseteq V$
is a real problem. We can solve this problem if $X$ is locally compact
(for the Scott topology) and coherent. Recall that a topological space
is \emph{coherent} if the intersection of any two compact saturated subsets
is compact. 

\begin{lem}\label{lem:H2'}
Suppose that $\phi=(\phi_1,\phi_2)\in \Hom_{0,1}(\A^X,\A)$ satisfies
the Heckmann condition (H2'). If  $X$ is a locally compact coherent
sober dcpo, then  $C=\cl(L)$ for the formal lens $(C,Q)$ and the real lens
$L=C\cap Q$ associated with $\phi$. 
\end{lem}

\begin{proof}
 Let $x\not\in \cl(L)$. We will show: there are an open neighborhood
 $W$ of $x$ with $W\cap L = \emptyset$ and an open set $U$ containing
 $Q$ such that $W\cap U\cap C=\emptyset$. 

We then can argue that $\phi_1(W\cap U)=0$. Since $\phi_2(U)=1$,
condition (H2') yields $\phi_1(W)= \phi_1(W\cap U)=0$, whence $W\cap
C=\emptyset$, which implies $x\not\in C$. 

Thus let $x\not\in \cl(L)$. By local compactness we can find a
compact saturated neighborhood $H$ of $x$ disjoint from $\cl(L)$.
By local compactness, not only every point but also every 
compact saturated subset $Q$ has a neighborhood basis of compact saturated
neighborhoods $(K_i)_i$. The intersections $H\cap K_i$ are compact 
by coherence, and they form a down-directed family such that
$\bigcap_i(H\cap K_i\cap C) \subseteq H\cap Q\cap C =H\cap L=
\emptyset$. Thus, $H \cap K_i\cap C = \emptyset$ for some $i$ by
\cite[Theorem II.1.21(3)]{dom}. If $W$ denotes the interior of $H$ and $U$ 
the interior of $K_i$, we have open sets with $W \cap U \cap C=\emptyset$. 
\end{proof} 

Each quasicontinuous dcpo (see \cite[Section III-3]{dom}) is locally
compact and sober according to \cite[Proposition III-3.7]{dom}. Thus,   
we can apply the previous results \ref{lem:H1'}, \ref{lem:H2'} to coherent
quasicontinuous dcpos. 

\begin{prop}
For coherent quasicontinuous dcpos, in particular for coherent continuous
dcpos, there is a one-to-one correspondence between (real) lenses and 
$\A$-valuations given by $L\mapsto\phi_L$ as in Lemma \ref{lem:lens}.
\end{prop}

The preceding  proposition covers most of the relevant cases since all 
FS-domains and, in particular, all retracts of bifinite domains are coherent. 
According to an unpublished result obtained independently by J. Goubault-Larrecq 
and A. Jung and by J. D. Lawson and Xi Xiaoyong, the generalizations of 
FS-domains and retracts of bifinite domains to QFS- and QRB-domains agree with
the compact coherent quasicontinuous dcpos, and our previous proposition applies 
to them, too. 
The result in the preceding Lemma can be strengthened by means of the
following lemma that is well known (see e.g. \cite[p. 370]{dom} or
\cite{abjudomain}): 

\begin{lem}
For every lens $L$ in a coherent sober dcpo, $\da L$ is closed. 
\end{lem}


This Lemma together with the previous one allows one to describe the
Plotkin powerdomain over a quasicontinuous coherent dcpo in the classical way 
as the collection of all (real) lenses $L$ with 
the Egli-Milner order $L\leq_{EM} L'$ iff $L\subseteq \da L'$ and
$L'\subseteq \ua L$ (see, e.g.\ \cite{abjudomain} or \cite[Theorem IV-8.18]{dom}). 

There is a second class of dcpos where the Plotkin powerdomain consists of
the (real) lenses, the countably based continuous dcpos  
(see \cite[]{abjudomain} or \cite[Theorem IV-8.18]{dom}). Also in this case the $\A$-valuations are in a
one-to-one correspondence with the (real) lenses. It would be
desirable to have a proof of the 
statement in Lemma \ref{lem:H2'} that would
cover both the countably based and the coherent case. 

\section{Final remarks }

We conclude with some comments on possible extensions of our results 
and limitations of our methods. 

\subsubsection*{Plotkin powerdomain via different computational
  prototypes}
We have seen that when choosing the computational prototype $\A$ 
as $\mathcal{P}(\Si)$, the Plotkin powerdomain of $\Si$, then $\A(X)$
will in general be different from $\mathcal{P}(X)$, the Plotkin powerdomain
of $X$. However, what happens if we choose $\A$ as $\mathcal{P}(A)$ for some
dcpo $A$ more complex than $\Si$, e.g.\ the domain $\mathbb{T}$ of lifted 
booleans? Does there exist an $A$ such that for $\A = \mathcal{P}(A)$ the
repletion $\A(X)$ always gives rise to $\mathcal{P}(X)$?

\subsubsection*{Combining probability with nondeterminism}
There are other algebraic effects where the methods developed above
can be applied. Battenfeld \cite{batisdt13} has looked at probabilistic
effects. He chooses the extended nonnegative reals $\ov{\mathbb R}_+$
with addition and multiplication by nonnegative scalars
$\lambda\in\mathbb R_+$ as computational prototype. Using known
results about the extended probabilistic powerdomain $\cV(X)$ 
(C. Jones \cite{jonesthesis}, R. Tix \cite{tixdiplom}) he shows
that $\cV(X)$ is the repletion of $X$ for any continuous dcpo
$X$. Thus, here again we have the phenomenon that the repletion agrees
with the free dcpo-algebra with respect to some natural equational
laws for probabilistic choice operators. 

There have been extensive investigations on combining probabilistic
and nondeterministic effects by Tix, Keimel, Plotkin
\cite{tikeplprobdom}, Mislove \cite{mislove}, Goubault-Larrecq
\cite{jglduality}. The free dcpo-algebras have been characterized 
from an equational point of view. As in Section
\ref{sec:examples} there is an angelic, a demonic and an erratic case
to consider. We conjecture that in the angelic and demonic cases the
observationally induced approach leads to the same result as the
equational approach, while there is a big gap between the two in the
erratic case as in the case of nondeterminism without probability
considered in Section \ref{sec:examples}.

\subsubsection*{Limitations}
In the example of nondeterministic effects considered in Section
\ref{sec:examples} we make use of the phenomenon that the semilattice
homomorphisms between two semilattices $A$ and $B$ form again a
semilattice $\Hom(A,B)$, the semilattice operation for homomorphisms
being defined pointwise. The repletion of a (continuous) dcpo $X$ was
always given by $\Hom_{0,1}(\uR^X,\uR)$, where $\uR$ was the
computational prototype. An analogous phenomenon occurs when dealing
with probabilities and when combining probability with
nondeterminism. This is a quite exceptional situation. For example, when
dealing with noncommutative monoids, the collection of all
homomorphisms between two monoids does not carry any natural monoid
structure. Thus, if our computational proptotype is a dcpo with a
continuous noncommutative monoid structure, the repletion cannot be
$\Hom(\uR^X,\uR)$. We will pursue this topic elsewhere.   

\subsubsection*{Predicate transformers}
 
For a dcpo $X$ a predicate is usually meant to be a continuous
function $u\colon X\to\2$ or, equivalently, an open subset $U$ of
$X$. If our computational prototyp is $\uR$, we will consider
\emph{$\uR$-valued predicates}, that is continuous functions $u\colon
X\to\uR$. Thus, $\uR^X$ is the dcpo of all $\uR$-valued predicates on
$X$. For dcpos $X$ and $Y$, a \emph{predicate transformer} will be a
continuous map $s\colon \uR^Y\to\uR^X$. 

Since in all our considerations the computational monads were kind of 
`submonads' of the continuation monad $R^{R^{(-)}}$ we consider
\emph{state transformers} to be continuous maps $t \colon X \to R^{R^Y}$ 
(transforming an input $x\in X$ to an output $t(x)\in R^{R^Y}$. 

Exponential transpose establishes a one-to-one
correspondence between state transformers and predicate transformers:
$$  (R^{R^Y})^X \cong (R^X)^{R^Y}.$$

If $\uR$ is a $\dcpo_\Omega$-algebra, $\Hom(\uR^X,\uR)$ is a subdcpo of
$\uR^{\uR^X}$. The state transformers $t\colon X\to \Hom(\uR^Y,\uR)$
correspond to the predicate transformers $s\in \Hom(\uR^Y,\uR^X)$
through exponential transpose:
$$\Hom(\uR^Y,\uR)^X\cong \Hom(\uR^Y,\uR^X)$$
We apply this to the particular cases of repletion considered 
in Section \ref{sec:examples}.\\ 

{\sc Angelic case}: For all dcpos $X$ and $Y$, the state transformers
$t\colon X\to 
\Hom_{0,1}(\2_\vee^Y,\2_\vee)$ are in one-to-one correspondence with
those predicate transformers $s\colon \2_\vee^Y\to\2_\vee^X$ that 
preserve binary join, bottom and top:
$$\Hom_{0,1}(\2_\vee^Y,\2_\vee)^X\cong \Hom_{0,1}(\2_\vee^Y,\2_\vee^X).$$
In terms of open sets, these predicate transformers $s$ are
characterized by the properties \[s(Y) =X,\ \ s(\emptyset)
= \emptyset,\ \ s(U\cup V) = s(U)\cup s(V).\]

{\sc Demonic case}: For all continuous dcpos $X$ and $Y$, the state
transformers $t\colon X\to 
\Hom_{0,1}(\2_\wedge^Y,\2_\wedge)$ are in one-to-one correspondence with
those predicate transformers $s\colon \2_\wedge^Y\to\2_\wedge^X$ that 
preserve binary meets (intersections), bottom and top:
$$\Hom_{0,1}(\2_\wedge^Y,\2_\wedge)^X\cong \Hom_{0,1}(\2_\wedge^Y,\2_\wedge^X).$$
In terms of open sets, these predicate transformers $s$ are
characterized by the properties \[s(Y) = X,\ \ s(\emptyset)
=\emptyset,\ \ s(U\cap V) = s(U)\cap s(V).\]

{\sc Erratic case}:  For all compact continuous dcpos $X$ and $Y$, 
the state transformers $t\colon X\to 
\Hom_{0,1}(\A^Y,\A)$ are in one-to-one correspondence with
those predicate transformers $s\colon \A^Y\to\A^X$ that 
preserve $\funion$, bottom and top:
$$\Hom_{0,1}(\A^Y,\A)^X\cong \Hom_{0,1}(\A^Y,\A^X).$$
While in the angelic and demonic case the predicates are $\2$-valued, this 
is not so in the erratic case. But we can represent an $\A$-valued predicate 
$u$ as a pair $u=(u_1,u_2)$ of $\2$-valued predicates with $u_1 \geq u_2$. 
Recall that the members $\phi\in \Hom_{0,1}(\A^Y,\A)$ can be 
represented as pairs $\phi=(\phi_1,\phi_2)\in
\Hom_{0,1}(\2_\vee^Y,\2_\vee) \times \Hom_{0,1}(\2_\wedge^Y,\2_\wedge)$
such that $\phi_1\geq \phi_2$. Consequently, the state transformers
$t\colon X\to \Hom_{0,1}(\A^Y,\A)$ can be seen as pairs $t=(t_1,t_2)$
of state transformers $t_1\colon X\to \Hom_{0,1}(\2_\vee^Y,\2_\vee)$
and $t_2\colon X\to \Hom_{0,1}(\2_\wedge^Y,\2_\wedge)$ such that
$t_1\geq t_2$. Thus, state transformers for the erratic case
consist of an angelic and a demonic state transformer, the angelic
one dominating the demonic one.  

Similarly, using Observation \ref{obs},
every semilattice homomorphism $s\colon  \A^Y\to\A^X$ can be seen as a
pair $s=(s_1,s_2)\in \Hom_{0,1}(\2_\vee^Y,\2_\vee^X)\times
\Hom_{0,1}(\2_\wedge^Y,\2_\wedge^X)$ such that $s_1\geq s_2$.
 Thus, predicate transformers for the erratic case
consist of an angelic and a demonic predicate transformer, the angelic
one dominating the demonic one. \\

{\sc Plotkin powerdomain}: The functional representation of the Plotkin 
powerdomain $\cP(X)$ for continuous dcpos $X$ through the Heckmann 
conditions allows one to characterize the predicate transformers 
$\A^Y\to\A^X$ corresponding to the state transformers $X\to\cP(Y)$
as the continuous maps  $s\colon \A^Y\to\A^X$ that satisfy:
\begin{enumerate}[label=(H\arabic*)]
\item[(H1)] if $s(u)(x)=\bot$ then $s(u\vee v)(x)=s(v)(x)$
\item[(H2)] if $s(u)(x)=\top$ then $s(u\wedge v)(x)=s(v)(x)$.
\end{enumerate}
If one represents $s$ as a pair $(s_1,s_2)$ of maps $\2^X \to \2$
these conditions read:
$s_1$ and $s_2$ are continuous maps preserving bottom, top and 
binary union, resp., intersection, and satisfy:
\begin{enumerate}[label=(H\arabic*')]
\item[(H1')] if $s_1(u)(x)=0$ then $s_2(u\vee v)(x)=s_2(v)(x)$
\item[(H2')] if $s_2(u)(x)=1$ then $s_1(u\wedge v)(x)=s_1(v)(x)$
\end{enumerate}
Thus, these predicate transformers consist of an angelic component $s_1$
and a demonic component $s_2$ which are related by the conditions (H1')
and (H2').

\bibliography{myrefs}
\bibliographystyle{plain}

\end{document}